\documentclass[leqno
,pdflatex
,prodmode
,acmtocl
]{acmsmall}

\usepackage{amsmath}
\usepackage{amssymb}
\usepackage{amsbsy}
\usepackage{xspace}

\usepackage{url}
\usepackage{multicol}
\usepackage{stmaryrd}

\usepackage{enumerate}

\usepackage{xspace}
\newcommand{\mathcmd}[1]{\ensuremath{#1}\xspace}

\newcommand{\dlfont}{\mathcal}
\newcommand{\dl}[1]{\mathcmd{\dlfont{#1}}}
\newcommand{\idRole}{\mathcmd{\textsf{\upshape\mdseries id}}}
\newcommand{\ALBOid}{\dl{ALBO\smash[t]{\mbox{${}^{\idRole}$}}}}
\newcommand{\ALC}{\dl{ALC}}
\def\And{\sqcap}
\def\Not{\neg}
\newcommand{\ALB}{\dl{ALB}}
\newcommand{\ALBO}{\dl{ALBO}}
\def\Or{\sqcup}
\newcommand{\subsumed}{\sqsubseteq}
\def\bigAnd{\bigsqcap}
\newcommand{\ALCO}{\dl{ALCO}}
\def\lAnd{\wedge}
\def\lOr{\vee}
\def\lNot{\neg}
\newcommand{\sub}{\mathcmd{\textsf{\upshape sub}}}

\newcommand{\N}{\mathbb{N}}

\newcommand{\metasmbfont}{\mathsf}
\newcommand{\Card}{\mathcmd{\metasmbfont{Card}}}
\makeatletter
\def\@define#1{%
    \mathcmd{%
        \stackrel%
            {\mbox{%
                \tiny\ensuremath{\metasmbfont{def}}%
                  }%
            }%
            {{#1}}%
            }%
              }
\def\sm@shdefine#1{%
\smash[t]{\@define{#1}}%
}
\newcommand{\define}{%
        \mathchoice{\@define{\ =\ }}{\sm@shdefine{=}}{\sm@shdefine{=}}{\sm@shdefine{=}}%
}
\newcommand{\defiff}{%
        \mathchoice%
            {\@define{\ \Longleftrightarrow\ }}%
            {\sm@shdefine{\Longleftrightarrow}}%
            {\sm@shdefine{\Leftrightarrow}}%
            {\sm@shdefine{\Leftrightarrow}}%
}
\makeatother
\newcommand{\ecl}[1]{\|#1\|}

\newcommand{\branch}[1]{\seg{#1}}
\newcommand{\seg}[1]{\mathcmd{\mathcal{#1}}}
\newcommand{\tor}{\,\mid\,}
\newcommand{\tand}{\,\,\ \ \,}
\makeatletter
\newcommand{\tableaulblfont}{\sffamily}
\iftagsleft@
\def\tableauleftlbldelim{}
\def\tableaurightlbldelim{:\ }
\else
\def\tableauleftlbldelim{\ :}
\def\tableaurightlbldelim{}
\fi
\def\tableaulblfmt#1{\text{\tableaulblfont\tableauleftlbldelim #1\tableaurightlbldelim}}
\def\@xtableaurule#1#2{%
\genfrac{}{}{}{0}{#1}{#2}%
}
\def\@ytableaurule#1#2[#3]{%
\def\lbl@{#3}
\iftagsleft@%
\tableaulblfmt{#3}\@xtableaurule{#1}{#2}%
\else%
\@xtableaurule{#1}{#2}\tableaulblfmt{#3}%
\fi%
\let\label\ltx@label
\let\@currentlabel\lbl@
}
\def\tableaurule#1#2{%
\@ifnextchar[{\@ytableaurule{#1}{#2}}{\@xtableaurule{#1}{#2}}%
}
\def\@xinlinetableaurule#1#2{%
{#1}/{#2}%
}
\def\@yinlinetableaurule#1#2[#3]{%
\def\lbl@{#3}
\iftagsleft@%
\tableaulblfmt{#3}\@xinlinetableaurule{#1}{#2}%
\else%
\@xinlinetableaurule{#1}{#2}\tableaulblfmt{#3}%
\fi%
\let\label\ltx@label
\let\@currentlabel\lbl@
}
\def\inlinetableaurule#1#2{%
\@ifnextchar[{\@yinlinetableaurule{#1}{#2}}{\@xinlinetableaurule{#1}{#2}}%
}
\newcount\@tskipcnt
\def\tfillsymbol{\mbox{\fontsize{3}{4}\selectfont.}}
\def\tfill{%
  \leavevmode
  \cleaders \hb@xt@ .44em{\hss{\tfillsymbol}\hss}\hfill
  \kern\z@}
\newcommand{\titem}[2]{$#1$\tfill #2}
\newcommand{\tbranch}{{\blacktriangleright}}
\newcommand{\unsat}{\text{Unsatisfiable}}
\newcommand{\sat}{\text{Satisfiable}}
\def\@@tskip{\phantom{\tbranch}}%
\def\@@@tskip{\@@tskip\advance\@tskipcnt\m@ne}%
\def\@tskip[#1]{%
\@tskipcnt#1
\loop\ifnum\@tskipcnt>\z@ \@@@tskip\repeat}%
\def\tskip{%
\@ifnextchar[{\@tskip}{\@@tskip}}
\makeatother
\def\complexityfont{\sffamily\upshape}
\def\complexity#1{\mathcmd{\text{\complexityfont #1}}}
\def\NExpTime{\complexity{NExpTime}}
\newcommand{\inlineproverfont}{\scshape\mdseries}
\DeclareTextFontCommand{\proverfont}{\inlineproverfont}
\makeatletter
\newcommand{\@mettel}{{\inlineproverfont Met\kern-.063emTeL}}
\DeclareRobustCommand{\mettel}{\mathcmd{\text{\@mettel}}}
\def\@mettelII{{\inlineproverfont\@mettel{}\kern-.33em\raise.7ex\hbox{2}}}
\def\f@smash{\ht\z@\z@ \box\z@}
\newcommand{\@@mettelII}{\makesm@sh{\@mettelII}\f@smash}
\DeclareRobustCommand{\mettelII}{\mathcmd{\text{\@@mettelII}}}
\makeatother

\newcommand{\spass}{\mathcmd{\proverfont{SPASS}}}
\newcommand{\mspass}{\mathcmd{\proverfont{MSPASS}}}
\newcommand{\vampire}{\mathcmd{\proverfont{Vampire}}}
\newcommand{\E}{\mathcmd{\proverfont{E}}}

\newcommand{\ALCQIb}{$\dlfont{ALCQI}b$\xspace}

\newcommand{\indiv}{a}
\newcommand{\cname}{A}
\newcommand{\rname}{Q}

\newcommand{\ctop}{\mathord{\top}}
\newcommand{\cbot}{\mathord{\bot}}
\newcommand{\roletop}{\mathord{\triangledown}}
\newcommand{\rolebot}{\mathord{\vartriangle}}

\newcommand{\Objects}{\textsf{\slshape O}\xspace\/}
\newcommand{\Concepts}{\textsf{\slshape C}\xspace\/}
\newcommand{\Roles}{\textsf{\slshape R}\xspace\/}

\newcommand{\DomR}{{\upharpoonright}}
\newcommand{\RngR}{{\downharpoonleft}}
\newcommand{\lcyl}[1]{\ensuremath{#1^{c}}}
\newcommand{\rcyl}[1]{\ensuremath{{}^{c}\!#1}}

\newcommand{\I}{\mathcal{I}}
\newcommand{\J}{\mathcal{J}}

\newcommand{\TALBOidub}{\mathcmd{T_\ALBOid\eqref{rule: unrestricted blocking}}}

\newcommand{\napplE}{\#^\exists}
\newcommand{\simB}{\mathop{\sim_\branch{B}}}
\newcommand{\notsimB}{\mathop{\not\sim_\branch{B}}}
\newcommand{\IB}{{\I(\branch{B})}}

\renewcommand{\tableauleftlbldelim}{(}
\renewcommand{\tableaurightlbldelim}{):\ }
\renewcommand{\tand}{,\ \ }
\renewcommand{\tbranch}{{\qquad\blacktriangleright}}

\newcommand{\ST}{\mathcmd{\textsf{\upshape ST}}}

\newlength{\mylength}
\title{Using Tableau to Decide Description Logics\\ with Full Role Negation and Identity}
\author{RENATE A. SCHMIDT and DMITRY TISHKOVSKY 
\affil{%
The University of Manchester, UK}}
\begin{abstract}
This paper presents a tableau approach for deciding expressive
description logics with full role negation and role identity.
We consider the description logic \ALBOid, which
is the extension of \ALC with the Boolean role operators,
inverse of roles, the identity role, and includes full support
for individuals and singleton concepts.
\ALBOid is expressively equivalent to the two-variable fragment of
first-order logic with equality and subsumes Boolean modal
logic.
In this paper we define a sound and complete tableau calculus for
the \ALBOid that provides a basis for
decision procedures for this logic and all its sublogics.
An important novelty of our approach is the use of a generic
unrestricted blocking mechanism.
Being based on a conceptually simple rule, unrestricted blocking 
performs case distinctions over whether two individuals are equal
or not and equality reasoning to find finite models.
The blocking mechanism ties the proof of termination of tableau
derivations to the %
finite model property of \ALBOid.
\end{abstract}

\acmVolume{}
\acmNumber{}
\acmArticle{}
\acmYear{}
\makeatletter
  \def\@journalName{}
  \def\@journalNameShort{}
  \def\@permissionCodeOne{}
  \def\firstfoot{}
  \ps@appheadings
\def\opentab{}
\def\closetab{}

\makeatother
\runningfoot{}
\begin{document}

\maketitle

\section{Introduction}

Mainstream description logics languages and ontology web languages
are equipped with a rich supply of syntactic constructs and operators
for supporting the needs of various applications.
In these languages the support of concepts and roles, the main
syntactic entities in description logic languages, is 
is slightly uneven however. A notable absence is the negation operator for roles.
In the description logic \ALC and other popular extensions
of \ALC it is possible to define
a \emph{spam filter} as a mechanism for filtering
out spam emails, and
a \emph{sound spam filter} as a spam filter that
filters out only spam emails,
by specifying
\begin{align*}
\textsf{spam-filter} & \define \textsf{mechanism}
\And \exists\textsf{filter-out}.\textsf{spam-email} \quad \text{and}
\\
\textsf{sound-spam-filter} & \define \textsf{spam-filter} \And
\Not\exists\textsf{filter-out}.\Not\textsf{spam-email}.
\end{align*}
It is not possible however to define \emph{a complete spam filter}
as a spam filter that filters out every spam email.
With role negation this can be expressed by the following. 
\begin{align*}
\textsf{complete-spam-filter} 
& \define \textsf{spam-filter} \And
\Not\exists\Not\textsf{filter-out}.\textsf{spam-email}.
\end{align*} 
The first occurrence of negation is a concept negation operator, which
almost all description logics support, while the second occurrence is a
role negation operator.
The role negation operator is
not available in \ALC or other current description logics that form
the basis of OWL
DL/1.1/2.0~\cite{BaaderCalvaneseEtal03,owl,owl2}.

The three examples can be expressed in first-order logic as follows.
\begin{align*}
\forall x [\textsf{spam-filter}(x) & \leftrightarrow .\
\textsf{mechanism}(x) \land \exists y [ \textsf{filter-out}(x,y)
\land \textsf{spam-email}(y)]]
\\
\tag{$\dagger$} 
\forall x [\textsf{sound-spam-filter}(x) & \leftrightarrow .\
\textsf{spam-filter}(x) \land \forall y [ \textsf{filter-out}(x,y) 
\rightarrow \textsf{spam-email}(y)]]
\\
\tag{$\ddagger$}
\forall x [\textsf{complete-spam-filter}(x) & \leftrightarrow .\
\textsf{spam-filter}(x) \land \forall y [ \textsf{spam-email}(y) 
\rightarrow  \textsf{filter-out}(x,y)]].
\end{align*} 
The right-hand-sides of the equivalences~($\dagger$) and ($\ddagger$) involve
universal quantification but of a different kind.
In~($\dagger$) it is the image elements of the role \textsf{filter-out}
that are universally quantified, while in ($\ddagger$) it is the
elements in the concept \textsf{spam-email} that are universally quantified.
($\dagger$)~expresses the \emph{necessity} of a property (what is
filtered out is necessarily a spam message), while ($\ddagger$)
expresses the \emph{sufficiency} of a property (being a spam message
is sufficient to be filtered out).
Natural examples of both kinds of universal quantification can
be found in many domains and are common in every-day language.

Motivated by pushing the limits of description logics, in this paper
we are interested in description logics that
allow role negation, and can therefore accommodate examples such as the
above,
but also provide a range of other role operators not typically
supported in mainstream description logics.
In particular, we focus on a description logic, called \ALBOid.
\ALBOid is the description logic \ALC extended by union
of roles, negation of roles, inverse of roles, and the identity role.
In addition, it provides full support for ABox individuals and singleton
concepts.
\ALBOid is an extension of the description logic
\ALB, introduced in~\cite{HustadtSchmidt00a}, with individuals and singleton concepts,
called nominals in modal logic, and the identity role.
\ALB is the extension of \ALC, in which concepts and roles form
Boolean algebras, and additional operators include inverse of roles
and a domain restriction operator.
It is shown in this paper (Section~\ref{section_definition_ALBO})
that the domain restriction operator
can be linearly encoded using the other operators of \ALB.

\ALBOid is a very expressive description logic.
It subsumes Boolean modal
logic~\cite{GargovPassy90,GargovPassyTinchev87} and tense, hybrid
versions of Boolean modal logic with the $@$~operator and
nominals.
It can also be shown that \ALBOid is expressively
equivalent to the two-variable fragment of first-order
logic with equality. %
\ALBOid is in fact very close to the brink of undecidability, because
it is known that adding role composition to \ALB
takes us into undecidable territory.

Since \ALBOid subsumes Boolean modal logic it follows
from~\citeN{LutzSattler02} that the satisfiability problem in \ALBOid
is \NExpTime-hard.
\citeN{GraedelKolaitisVardi97} showed that satisfiability in the
two-variable first-order fragment with equality is \NExpTime-complete.
It follows therefore that the computational complexity of
\ALBOid-satisfiability is \NExpTime-complete.

Description logics with full role negation can
be decided by translation to first-order logic and first-order
resolution theorem provers such as \mspass~\cite{HustadtSchmidtWeidenbach99}, 
\spass~\cite{WeidenbachSchmidtEtAl07}, \E~\cite{Schulz02} and
\vampire~\cite{RiazanovVoronkov99}.
The paper~\cite{HustadtSchmidt00a} shows that the logic \ALB can be
decided by translation to first-order logic and ordered resolution.
This result is extended by~\citeN{DeNivelleSchmidtHustadt00} to \ALB
with positive occurrences of composition of roles.
\ALBOid can be embedded into the two-variable fragment of first-order
logic with equality which can be decided with first-order resolution
methods~\cite{DeNivellePrattHartmann01}.
This means that \ALBOid %
can be decided using first-order
resolution methods.

None of the current tableau-based description
logic systems are able to handle \ALBOid or \ALB because they
do not support full role negation.
In fact, few tableau calculi or tableau procedures have been described for
description logics with complex role operators, or equivalent dynamic
modal logic versions.
Ground semantic tableau calculi and tableau decision procedures
are presented by~\citeN{DeNivelleSchmidtHustadt00} for the modal
versions of $\ALC(\Or,\And,{}^{-1})$, that is, \ALC with role union, role
intersection and role inverse. These are extended with the domain restriction
operator to $\ALC(\Or,\And,{}^{-1}, \DomR)$ by~\citeN{Schmidt06b}.
A semantic tableau decision procedure for \ALC with role intersection,
role inverse, role identity and role composition is described
by~\citeN{Massacci01}.
None of these tableaux make provision for the role negation operator however.

A tableau decision procedure for the description logic \ALCQIb which
allows for Boolean combinations of `safe' occurrences of negated
roles is discussed by~\citeN{Tobies01}.
From the clausal form of the first-order translation it can be seen that
safeness ensures guardedness which is
violated by unsafe occurrences of role negation.
Guardedness is also important in the definition of a generalisation 
with relational properties of the dynamic modal logic corresponding to
$\ALC(\Or,\And,{}^{-1})$ in~\citeN{DeNivelleSchmidtHustadt00}, where it
is shown that hyperresolution and splitting decides this logic.
The hyperresolution decision procedure for this logic is essentially a
hypertableau approach with similarities to~\citeN{MotikShearerHorrocks09}.
Decidability of propositional dynamic logic with negation of atomic
relations using B\"uchi automata is shown by~\citeN{LutzWalther04}.
\citeN{SchmidtOrlowskaHustadt04a} presented a sound and complete ground
semantic tableau calculus for Peirce logic, %
which is equivalent to the extension of \ALB with role composition and
role identity.
However the tableau is not terminating because reasoning in Peirce
logic is not decidable.

In this paper we present a ground semantic tableau approach that
decides the description logic \ALBOid.
In order to limit the number of individuals in the tableau derivation and guarantee
termination we need a
mechanism for finding finite models.
\emph{Standard loop checking} mechanisms are based on
comparing sets of (labelled or unlabelled)
concept expressions such as subset blocking or equality blocking
in order to detect periodicity in the underlying models.
Instead of using the standard loop checking mechanisms,
our approach is based on a new inference rule, called the \emph{unrestricted
blocking} rule, and equality reasoning. 
Our approach has the following advantages over standard loop checking.
\begin{longitem}
 \item It is conceptually simple and easy to implement.
 It does not require specialised blocking tests and 
procedural descriptions in terms of status variables.
All individuals are blockable and once blocked remain blocked.
 \item The blocking mechanism is generally sound and complete. 
    This means that it is applicable also to other deduction
    methods and other
    logics~\cite{SchmidtTishkovsky-GTM+-2008}
    including full first-order logic, where it can be used
    to find finite models. 
 \item It provides greater flexibility in constructing models.
       For instance, it can be used to construct small models for a satisfiable concept,
       including domain minimal models.
 \item
The approach has the advantage that it constructs
real models, whereas existing tableau procedures for many OWL 
description logics construct only pseudo-models that are not always
real models but can be completed by post-processing to real
models (which may be infinite).
\item
It can be implemented and simulated in first-order logic
provers~\cite{BaumgartnerSchmidt08}.
\end{longitem}

The style of presentation of our tableau calculus is similar to
presentations in modal and hybrid logic,
for example~\cite{Fitting-TMF+-1972,Blackburn-ILD-2000,DeNivelleSchmidtHustadt00,Schmidt06b,SchmidtOrlowskaHustadt04a,delCerroGasquet-GF+-2002}.
Notable about our calculus compared, for example,
with~\cite{Schmidt06b,SchmidtOrlowskaHustadt04a} is that it operates
only on ground labelled \emph{concept} expressions.
This makes it easier in principle to implement the calculus as extensions
of tableau-based description logic systems that can handle
singleton concepts and include equality reasoning.

The structure of the paper is as follows.
The syntax and semantics of \ALBOid are defined in
Section~\ref{section_definition_ALBO}.
In Section~\ref{section: efmp} we prove that \ALBOid has
the effective finite model property by reducing \ALBOid-satisfiability to
the two-variable fragment of the first-order logic.
We define a tableau calculus for \ALBOid in Section~\ref{section_tableau},
and in Section~\ref{section: Soundness and completeness}, we prove that it is sound and complete without the unrestricted
blocking rule.
Sections~\ref{section_blocking} and~\ref{section_termination} introduce
our blocking mechanism and prove soundness, completeness and termination of the
tableau calculus extended with this mechanism.
In Section~\ref{section: Complexity}
we prove that an appropriate strategy for the application of the unrestricted
blocking rule yields a tableau procedure with optimal complexity.
We define general criteria for deterministic decision
procedures for \ALBOid (and any sublogics) in
Section~\ref{section_implementation}.
Various issues concerning blocking and our results are discussed in
Section~\ref{section_discussion}.

The paper is an extended and improved version of the conference
paper~\cite{SchmidtTishkovsky-UTD+-2007} and builds on results in
\citeN{SchmidtTishkovsky-GTM+-2008}.

\section{Syntax and Semantics of \protect\ALBOid}
\label{section_definition_ALBO}

The syntax of \ALBOid is defined over a signature, denoted by
$\Sigma=(\Objects,\Concepts,\Roles)$, consisting of three disjoint alphabets:
$\Objects$, the alphabet of individuals or object names,
$\Concepts$, the alphabet of concept symbols,
and $\Roles$, the alphabet of role symbols.
For individuals we use the notation $\indiv, \indiv', \indiv'', \indiv_i$, 
for concept symbols $\cname, \cname', \cname_i$, and 
for role symbols $\rname, \rname', \rname'', \rname_i$. 
The language includes a special role symbol~$\idRole$ for
the \emph{identity} role which is regarded as a constant.
The logical connectives are:
$\Not$ (\emph{negation}), $\Or$ (\emph{union} for both concepts and roles), $\exists$
(\emph{existential concept restriction}), and ${}^{-1}$ (role \emph{inverse}).
\emph{Concept expressions} (or \emph{concepts}) $C, D$ and \emph{role
expressions}
(or \emph{roles}) $R, S$ are defined by the following grammar rules:
\begin{align*}
 C, D &\define \cname\ \mid\ \{\indiv\}\ \mid\ \Not C\ \mid\ C\Or D\ \mid\ \exists R.C,\\
 R, S &\define \rname\ \mid\ R\Or S\ \mid\ R^{-1}\ \mid\ \Not R\ \mid\ \idRole.
\end{align*}
We refer to $\{\indiv\}$ as a %
\emph{singleton}
concept.
By the \emph{length} of an \ALBOid-expression $E$ we understand
the length in symbols of a word which represents $E$
in the above language.

A TBox is a finite set of concept inclusion statements of the
form $C\subsumed D$ and an RBox is a finite set of
role inclusion statements of the form $R\subsumed S$. 
An ABox is a finite set of statements of the form $\indiv:C$
and $(\indiv,\indiv'):R$, which are called \emph{concept assertion} and 
\emph{role assertions}, respectively.
A \emph{knowledge base} is a tuple~$(\mathcal{T},\mathcal{R},\mathcal{A})$
consisting of a TBox $\mathcal{T}$, an RBox $\mathcal{R}$,
and an ABox $\mathcal{A}$.

Next, we define the semantics of \ALBOid.
A \emph{model} (or an \emph{interpretation}) $\I$ of \ALBOid is a tuple
\[\I=(\Delta^\I,\cname^\I,\ldots,\indiv^\I,\ldots,\rname^\I,\ldots),\]
where $\Delta^\I$ is a non-empty set, and for any concept symbol $\cname$, any role
symbol $\rname$ and any individual~$\indiv$: 
\[
\text{$\cname^\I$ is a subset of $\Delta^\I$,} \quad
\indiv^\I\in\Delta^\I \quad \text{and}\quad
\text{$\rname^\I$ is a binary relation over $\Delta^\I$}.
\]
We expand the interpretation $\I$ to all concepts and roles with the use of the following definitions
(for
any concepts $C$, $D$, roles $R$, $S$, and individuals~$\indiv$):
\begin{align*}
 \{\indiv\}^\I &\define\{\indiv^\I\},&
 (R\Or S)^\I &\define R^\I\cup S^\I,\\
 (\Not C)^\I &\define\Delta^\I\setminus C^\I,&
 (R^{-1})^\I &\define(R^\I)^{-1}\ =\ \{(x,y)\mid(y,x)\in R^\I\},\\
 (C\Or D)^\I &\define C^\I\cup D^\I,&
 (\Not R)^\I &\define(\Delta^\I\times\Delta^\I)\setminus R^\I,\\
 (\exists R.C)^\I &\define\{x\mid\exists y\in C^\I\ (x,y)\in R^\I\},&
 \idRole^\I&\define\{(x,x)\mid x\in\Delta^\I\}.
\end{align*}

A concept $C$ is \emph{satisfied} in a model $\I$ iff
$C^\I\neq\emptyset$.
A set of concepts is satisfied in a model $\I$ iff
all the concepts in the set are satisfied in $\I$.
A concept is \emph{satisfiable} %
iff there is a model where it is satisfied.

Slightly departing from description logic conventions, we say 
a concept $C$ is \emph{valid} in a model $\I$ iff
$C^\I=\Delta^\I$.
A set of concepts is valid in a model $\I$ iff
all the concepts in the set are valid in $\I$.
If $E$ is a concept expression, a concept inclusion statement, a role
inclusion statement, a concept assertion or a role assertion,
we indicate by $\I\models E$ that~$E$ is valid in the model $\I$.\footnote{In the literature $\I\models E$ is
normally only defined when $E$ is an inclusion, equivalence or
assertion, and then it is said that $E$ is true in $\I$.}
In particular, shortening the definition of concept validity, we can write that,
in any model $\I$ and for any concept $C$, 
\[\I\models C\defiff C^\I=\Delta^\I.\]
Concept and role inclusion statements are interpreted as subset relationships, and
concept and role assertions are interpreted as element-of
relationships. In particular, in any model $\I$ and
for any concepts~$C$,~$D$, any roles $R$, $S$, and any individuals $\indiv$, $\indiv'$:
\begin{align*}
\I \models C\subsumed D& \defiff C^\I\subseteq D^\I,&
\I \models \indiv:C& \defiff \indiv^\I\in C^\I,\\
\I \models R\subsumed S& \defiff R^\I\subseteq S^\I,&
\I \models (\indiv,\indiv'):R& \defiff (\indiv^\I,\indiv'^\I)\in R^\I.
\end{align*}

A concept $C$ is \emph{satisfiable with respect to a knowledge base
$(\mathcal{T},\mathcal{R},\mathcal{A})$} iff $C$ is satisfied in a
model~$\I$ validating all the statements from the knowledge base,
that is, $\I\models E$ for every $E\in \mathcal{T}\cup \mathcal{R}\cup
\mathcal{A}$.

\ALBOid has considerable expressive power
and many useful additional operators can be defined in it.
These include the following standard operators:
\begin{align*}
& \text{Top concept:} &\ctop & \define \cname\Or\Not \cname\text{ (for some concept symbol $\cname$)}\\
& \text{Bottom concept:} &\cbot & \define \Not\ctop\\
& \text{Concept intersection:} & C \And D& \define \Not(\Not C\Or\Not D)\\
& \text{Universal restriction:}& \forall R.C & \define \Not\exists R.\Not C\\
& \text{Image operator:} & \exists^{-1}R.C & \define \exists R^{-1}.C\\
\intertext{
plus these operators, which can be defined using role negation:
}
& \text{Top/universal role:} &\roletop & \define \rname\Or\Not \rname \text{ (for some role symbol $\rname$)}\\
& \text{Bottom/empty role:} &\rolebot & \define \Not\roletop \\
& \text{Role intersection:} & R \And S& \define \Not(\Not R\Or\Not S)\\
& \text{Diversity role:} &\textsf{div} & \define \Not\idRole \\
& \text{Universal modality:} &\Box C & \define \forall\roletop.C\\
& \makebox[\mylength][l]{\text{Sufficiency operator:}} & 
    \makebox[.3\mylength][r]{$\overline\forall R.C$} & \define \makebox[1.8\mylength][l]{$\neg \exists\neg R.C$}
\end{align*}
The bottom role can be expressed by using a concept inclusion stating
that the domain of the role is empty.
The sufficiency operator is also known as the window operator, see for example~\cite{GargovPassyTinchev87}.

Concept assertions can be internalised as concept expressions
as follows:
\begin{align*}
& \text{Concept assertions:}&
\indiv:C & \define\exists\roletop.(\{\indiv\}\And C).
\intertext{%
A role assertion $(\indiv,\indiv'):R$ can be expressed as a
concept assertion, namely}
& \text{Role assertions:}&
(\indiv,\indiv'):R & \define \indiv:\exists R.\{\indiv'\},
\intertext{%
or, using the above, it can be internalised as a concept expression.
In addition, concept and role inclusion axioms can be internalised
as concept expressions.
}
& \makebox[\mylength][l]{\text{Concept inclusion:}}&
    \makebox[.3\mylength][r]{$C\subsumed D$}& 
    \define \makebox[1.8\mylength][l]{$\Box(\Not C\Or D)$}\\
& \text{Role inclusion:}&
R\subsumed S & \define \Box\forall\Not(\Not R\Or S).\cbot
\end{align*}

Concept assertions, role assertions, concept inclusions and role
inclusions in \ALBOid are therefore all representable as concept
expressions.
This means that in \ALBOid the distinction between TBoxes, ABoxes and
RBoxes is strictly no longer necessary, and concept satisfiability in
\ALBOid with respect to any knowledge base can be reduced to concept
satisfiability with respect to a knowledge base where the TBox, the RBox,
and the ABox are all empty.
More precisely, we have that a concept $C$ is satisfiable with respect to a knowledge
base $(\mathcal{T},\mathcal{R},\mathcal{A})$ iff\label{page: eliminating knowledge base} the conjunction 
\[
C\And\bigAnd (\mathcal{T}\cup\mathcal{R}\cup\mathcal{A})
\]
of the concept
$C$ and all the concept expressions from the set $\mathcal{T}\cup\mathcal{R}\cup\mathcal{A}$
is satisfiable. 
Similarly, a set $\mathcal{S}$ of concepts is satisfiable with
respect to a knowledge base  $(\mathcal{T},\mathcal{R},\mathcal{A})$ iff the set $\mathcal{S}\cup
\mathcal{T}\cup\mathcal{R}\cup\mathcal{A}$ is satisfiable.
Without loss of generality, in this paper, we therefore focus on the
problem of concept satisfiability in \ALBOid.

Boolean combinations of inclusion and assertion statements
of concepts and roles are expressible in \ALBOid, as the
corresponding Boolean combinations of the concept expressions which represent
these statements.
In fact,  
we are allowed to use such statements in concept and role expressions
within the language of \ALBOid provided that 
a well-formed \ALBOid expression is obtained
after replacing  
all these statements 
in the original expression
for well-formed \ALBOid expressions which define them.
For example, the following expression is a well-formed concept expression in \ALBOid:
\begin{eqnarray*}
& (\left(\textsf{Alice}\!:\!\Not\exists\textsf{likes}.\textsf{football\_fan}\right)\And
\left(\textsf{Bob}\!:\!\textsf{football\_fan}\right)\And
\left(\textsf{hasFriend}\subsumed\textsf{likes}\right))\\
& \subsumed
((\textsf{Alice},\textsf{Bob})\!:\!\Not\textsf{hasFriend}).
\end{eqnarray*}
It says that Alice and Bob cannot be friends because Bob is fanatical
about football and Alice does not like football fans.
The sentence uses an assumption expressed as a role inclusion that
somebody's friend is necessarily a person who he or she likes.
Notice that the role inclusion statement denotes an \ALBOid
expression which essentially uses the role negation operator.

The encodings of the operators mentioned above
preserve logical equivalence. 
Additionally,
\ALBOid also has enough expressive power to
represent concept satisfiability problems 
which involve the operators defined below: 
\begin{align*}
& \text{Test operator:} & (C\text{?})^\I & \define \{(x,x)\in\Delta^\I\times\Delta^\I\mid x\in C^\I\},\\
& \makebox[\mylength][l]{\text{Domain restriction:}} &
    \makebox[.3\mylength][r]{$(R \DomR C)^\I$} &
    \define \makebox[1.8\mylength][l]{$\{(x,y)\in R\mid x\in C^\I\},$}\\
& \text{Range restriction:} &
 (R \RngR C)^\I &\define\{(x,y)\in R\mid y\in C^\I\}, \\
& \text{Left cylindrification:} &
(\lcyl{D})^\I &\define\{(x,y)\in\Delta^\I\times\Delta^\I\mid x\in D^\I\},\\    
& \text{Right cylindrification:} &
 (\rcyl{D})^\I &\define\{(x,y)\in\Delta^\I\times\Delta^\I\mid y\in D^\I\},\\
& \text{Cross product:} &
 (C\times D)^\I &\define C^\I\times D^\I\ =\ \{(x,y)\mid x\in C^\I, y\in D^\I\}.
\end{align*}
The test operator can be represented using the range restriction operator and the identity role:
$(C\text{?})^\I=(\idRole\RngR C)^\I$.
In the presence of role inverse, the universal role and role intersection
any single operator from $\{\DomR,
\RngR, {\cdot}^c, {}^c\!{\cdot}, \times\}$ is enough to define the
remaining operators from this set. 
In any model $\I$, the following equalities are
true:
\begin{align*}
(R \DomR C)^\I &\ =\ \left((R^{-1} \RngR C)^{-1}\right)^\I &
(R \RngR C)^\I &\ =\ \left((R^{-1} \DomR C)^{-1}\right)^\I \\ 
(\lcyl{D})^\I &\ =\ \left((\rcyl{D})^{-1}\right)^\I & 
(\rcyl{D})^\I &\ =\ \left(\lcyl{D}{}^{-1}\right)^\I \\
(\lcyl{D})^\I &\ =\ (\roletop\DomR D)^\I&
(\rcyl{D})^\I &\ =\ (\roletop\RngR D)^\I\\
(R \DomR D)^\I &\ =\ (R\And D^c)^\I&
(R \RngR D)^\I &\ =\ (R\And{}^c\!{D})^\I\\ 
(\lcyl{D})^\I &\ =\ (D\times\top)^\I&
(\rcyl{D})^\I &\ =\ (\top\times D)^\I\\ 
(C\times D)^\I &\ =\ (\lcyl{C}\And \rcyl{D})^\I &
\end{align*}
Now it is enough to show that one of these operators can be encoded in \ALBOid.
For example, expressions involving left cylindrification can be
\emph{linearly} encoded in \ALBOid by replacing all occurrences of
$\lcyl{D}$ in $C$ by a new role symbol $\rname_{D}$ uniquely associated with $D$
and adding the definitions
$\Not D\subsumed\forall \rname_{D}.\bot$ and
$D\subsumed\Not\exists\Not \rname_{D}.\top$
to the knowledge base.
This encoding preserves satisfiability equivalence.
In a similar way, the other operators from 
$\{\DomR,\RngR, {\cdot}^c, {}^c\!{\cdot}, \times\}$ can be linearly encoded in \ALBOid.

Often description logics are required to satisfy the unique name
assumption.
We do not assume it for \ALBOid.
However, the
unique name assumption can be enforced by adding disjointness
statements of the form $\{\indiv\} \And \{\indiv'\} \subsumed \bot$, for
every distinct pair of individuals that occur in the given knowledge base, to the TBox.

In this paper we do not use role assertions in the form
$(\indiv,\indiv'):R$ but use them in the form $\indiv:\exists R.\{\indiv'\}$.
We refer to $\indiv:\exists R.\{\indiv'\}$
as a \emph{link} (between the individuals $\indiv$
and $\indiv'$).

Later we refer to the description logics $\ALCO$ and $\ALCO^{\Box}$.
\ALCO is a sublogic of \ALBOid without all the role operators and the identity role.
$\ALCO^{\Box}$ extends \ALCO with the universal modality operator~$\Box$.

\section{Effective Finite Model Property}
\label{section: efmp}

In this section we show that \ALBOid has the effective finite model
property. In order to make use of the well-known result of~\citeN{GraedelKolaitisVardi97}
that the two-variable fragment of the first-order logic with equality
has the finite model property,
we encode \ALBOid-satisfiability of concepts in the
two-variable fragment of the first-order logic with equality.
Additionally, we define a mapping of models for the two-variable fragment back to \ALBOid-models
and prove that it preserves satisfiability modulo the mentioned encoding.

For every concept symbol $\cname$, role symbol $\rname$, and individual $\indiv$,
let $\underline{\cname}$ be a unary predicate symbol, $\underline{\rname}$
a binary predicate symbol, and $\underline{\indiv}$  a constant which
are uniquely associated with $\cname$, $\rname$, and $\indiv$, respectively.
The following defines a mapping $\ST$ as the standard translation of \ALBOid
parametrised by the variables $x$ and~$y$.
The standard translation of concepts is:
\begin{align*}%
 \ST_x(\cname) & \define \underline{\cname}(x) & \ST_x(\{\indiv\}) &\define x\approx\underline{\indiv}\\
 \ST_x(\Not C) &\define \lNot\ST_x(C) & \ST_x(C\Or D) &\define \ST_x(C)\lOr \ST_x(D)\\
 \ST_x(\exists R.C) &\define\exists y\,(\ST_{xy}(R)\lAnd \ST_y(C))\\
\intertext{The standard translation of roles is:}
 \ST_{xy}(\rname) & \define \underline{\rname}(x,y) & \ST_{xy}(\idRole) &\define x\approx y\\
 \ST_{xy}(R^{-1}) &\define \ST_{yx}(R) & \ST_{xy}(R\Or S) &\define \ST_{xy}(R)\lOr \ST_{xy}(S)\\
 \ST_{xy}(\Not R) &\define \lNot\ST_{xy}(R)
\end{align*}
Clearly, only two variables occur in $\ST_x(C)$ for every concept~$C$.
Given a concept $C$, the translation $\ST_x(C)$ of~$C$
can be computed in linear time with respect to the length of~$C$.

As the standard translation just follows the definition of the semantics
of \ALBOid, every \ALBOid-model~$\I$ can be viewed as a first-order model
over the signature $\underline{\Sigma}=(\{\underline{\indiv}\mid\indiv\in\Objects\},
                                       \{\underline{\cname}\mid \cname\in\Concepts\},
                                       \{\underline{\rname}\mid \rname\in\Roles\})$,
where the interpretations of the symbols
$\underline{\cname}$, $\underline{\rname}$, and $\underline{\indiv}$
are $\cname^\I$, $\rname^\I$, and $\indiv^\I$ respectively.
Similarly, every first-order model over a signature $\underline{\Sigma}$
can be seen as an \ALBOid-model.

We write $\I\models\phi[x_1\mapsto a_1,\ldots,x_n\mapsto a_n]$
to indicate that all free variables of 
a first-order formula~$\phi$ 
are contained in the set $\{x_1,\ldots,x_n\}$
and $\phi$ is true in $\I$ where the variables 
$x_1,\ldots,x_n$ are assigned the domain elements $a_1,\ldots,a_n$, respectively.

The following theorem can be proved by simultaneous induction on the
lengths of concepts and roles.
\begin{theorem}\label{theorem: embedding into FO2}
Let $\I$ be any \ALBOid-model. Then, for any concept $C$, any role $R$ and any $a,b\in\Delta^\I$,
\begin{itemize}
 \item $a\in C^\I$ iff $\I\models\ST_x(C)[x\mapsto a]$, and
 \item $(a,b)\in R^\I$ iff $\I\models\ST_{xy}(R)[x\mapsto a,y\mapsto
b]$.
\end{itemize}
\end{theorem}

The effective finite model property
 for description logics can be stated as follows.
 A description logic~$L$ has the \emph{effective finite model property} iff
 there is a computable 
 function $\mu:\N\to\N$,
 such that the following holds.
    \begin{itemize}
     \setlength{\topskip}{0pt}
     \setlength{\itemsep}{0pt}
     \setlength{\itemindent}{0pt}
     \item[]
           For every concept $C$,  
           if $C$ is satisfiable in an $L$-model
           then there is a finite $L$-model for $C$ 
           with the number of elements in the domain not exceeding $\mu(n)$,
           where $n$ is the length of $C$. %
     \end{itemize}
We call $\mu$ the \emph{(model) bounding function} for $L$.
 
 The notion of the effective finite model property for 
 fragments of the first-order logic can be obtained
 from the above definition replacing the phrases `logic $L$'
 and `concept~$C$' with `fragment $\mathcal{F}$' and `formula $\phi$', respectively.

The effective finite model property is sometimes called
the bounded model property or the small model property
and is stronger than the finite
model property.
A logic has the finite model property if any satisfiable concept (formula) has
a finite model. 
With the effective finite model property there is a bound on the size
of the model that is given a priori based on the given concept (or formula).
This means there is a naive procedure for deciding satisfiability,
which is based on successively computing larger and larger models within
the bound and model checking the given concept with respect to the obtained models. 

The effective finite model property for the two-variable fragment 
of the first-order logic with equality is proved by~\citeN{Mortimer75}.
An exponential model bounding function $\mu$ is given
by~\citeN{GraedelKolaitisVardi97}. 
Because we will use this function in Sections~\ref{section_termination} and~\ref{section: Complexity},
we define it explicitly:
\[\mu(n)\define3(n\lfloor\log(n+1)\rfloor)\cdot 2^n.\]

\begin{theorem}[\citeN{GraedelKolaitisVardi97}]\label{theorem: efmp for FO2}
The two variable fragment of the first order logic with equality
has the effective finite model property with the model bounding
function $\mu(n)$.
\end{theorem}

Theorems~\ref{theorem: embedding into FO2} and~\ref{theorem: efmp for FO2}
imply that \ALBOid has the effective finite model property:
\begin{theorem}[Effective Finite Model Property of \ALBOid]\label{theorem: EFMP}
\ALBOid has the effective finite model property
with the model bounding function $\mu(n)$.
\end{theorem}
This provides an upper bound for the worst-case complexity of testing 
concept satisfiability
for \ALBOid. 
The complexity result for Boolean modal
logic of~\citeN{LutzSattler02} gives the needed lower bound for concluding:
\begin{theorem}
The concept satisfiability problem for \ALBOid is \NExpTime-complete.
\end{theorem}

\section{Tableau Calculus}
\label{section_tableau}

Let $T$ denote a tableau calculus comprising of a set of inference
rules. 
A \emph{derivation} or \emph{tableau} for $T$ is a finitely branching,
ordered tree whose nodes are sets of labelled concept expressions.
Assuming that~$\mathcal S$ is the input set of concept expressions
to be tested for satisfiability, the root node of the tableau is the
set $\{\indiv:C \mid C \in \mathcal S\}$, where $\indiv$ denotes a fresh individual.
Successor nodes are constructed in accordance with a set of inference
rules in the calculus.
The inference rules have the general form 
\[
\tableaurule{X_0}{X_1 \tor \ldots \tor X_n},
\]
where $X_0$ is the set of premises and the $X_i$
are the sets of conclusions.
If $n=0$, the rule is called \emph{closure rule} and written 
$\inlinetableaurule{X_0}{\bot}$.
An inference rule is applicable to a selected labelled concept
expression~$E$ in a
node of the tableau, if $E$ together with possibly other
labelled concept expressions in the node, are simultaneous
instantiations of all the premises of the rule.
Then $n$~successor nodes are created which contain the formulae of
the current node and the appropriate instances of $X_1, \ldots, X_n$.
We assume that \emph{any rule is applied at most once to the same set
of premises}, which is a standard assumption for tableau derivations.

We use the notation $T(\mathcal S)$ for a fully expanded tableau built
by applying the rules of the calculus~$T$ starting with the set
$\mathcal S$ as input.
That is, we assume that all branches in the tableau are fully expanded and
all applicable rules have been applied in $T(\mathcal S)$.
For any concepts $C_1,\ldots, C_n$, the notation $T(C_1,\ldots,C_n)$
is used instead of $T(\{C_1,\ldots,C_n\})$.

In a tableau, a maximal path from the root node is called a
\emph{branch}.
For $\branch{B}$ a branch of a tableau we write $D\in\branch{B}$
to indicate that the concept $D$ has
been derived in $\branch{B}$, that is, $D$ belongs to a node of the
branch~$\branch{B}$.
The notion of a tableau branch we use in this paper can be viewed
in two ways.
On one hand it has a procedural flavour as a path of nodes in the
tableau.
On the other hand a branch can be identified with the
set-theoretical union of the nodes in it.
A more careful distinction between these perspectives leads to the classical
notion of a Hintikka set but this is not essential for the paper.

A branch of a tableau is \emph{closed} if a closure rule has been applied
in this branch, otherwise the branch is called \emph{open}.
Clearly, expansion of any closed branch can be stopped immediately after
the first application of a closure rule in the branch.
The tableau $T(\mathcal{S})$ is \emph{closed} if all its branches are
closed and $T(\mathcal{S})$ is \emph{open} otherwise.
The calculus $T$ is \emph{sound} iff
for any (possibly infinite)
set of concepts
$\mathcal{S}$, each $T(\mathcal{S})$ is open whenever $\mathcal{S}$ is satisfiable.
$T$ is \emph{complete} iff for any (possibly infinite) unsatisfiable set of concepts
$\mathcal{S}$ there is a tableau $T(\mathcal{S})$ which is closed.
$T$ is said to be \emph{terminating} (for satisfiability) iff 
for every \emph{finite} set of concepts~$\mathcal{S}$
every closed tableau $T(\mathcal{S})$ is finite and
every open tableau $T(\mathcal{S})$ has a finite open branch.

\begin{figure}[!tb]
\begin{align*}
\intertext{Rules for \ALCO:}
&\tableaurule{\indiv: C\tand \indiv:\Not C}{\bot}[$\bot$]\label{rule: clash}
&&
\tableaurule{\indiv:\Not\Not C}{\indiv:C}[$\Not\Not$]\label{rule: not not}
\\
&\tableaurule{\indiv:\Not(C\Or D)}{\indiv:\Not C\tand \indiv:\Not D}[$\Not\Or$]\label{rule: not or}
&&
\tableaurule{\indiv:(C\Or D)}{\indiv:C\tor \indiv:D}[$\Or$]\label{rule: or}
\\
&\tableaurule{\indiv:\exists R.C}{\indiv:\exists R.\{\indiv'\}\tand \indiv':C}[$\exists$]%
\label{rule: exists} \ \text{($\indiv'$ is new)}
&&
\tableaurule{\indiv:\Not\exists R.C\tand \indiv:\exists R.\{\indiv'\}}{\indiv':\Not C}[$\Not\exists$]\label{rule: not exists}
\\
&\tableaurule{\indiv:\{\indiv'\}}{\indiv':\{\indiv\}}[sym]\label{rule: sym}
&&
\tableaurule{\indiv:\Not\{\indiv'\}}{\indiv':\Not\{\indiv\}}[$\neg$sym]\label{rule: nsym}
\\
&
\tableaurule{\indiv:\{\indiv'\}\tand \indiv':C}{\indiv:C}[mon]\label{rule: mon}
&&
\tableaurule{\indiv:C}{\indiv:\{\indiv\}}[refl]\label{rule: id}
\intertext{Rules for complex roles:}
&\tableaurule{\indiv:\exists (R\Or S).\{\indiv'\}}{\indiv:\exists R.\{\indiv'\}\tor \indiv:\exists S.\{\indiv'\}}[$\exists\Or$]%
\label{rule: exists or}
&&
\tableaurule{\indiv:\Not\exists (R\Or S).C}{\indiv:\Not\exists R.C\tand\indiv:\Not\exists S.C}[$\Not\exists\Or$]%
\label{rule: not exists or}
\\
&\tableaurule{\indiv:\exists R^{-1}.\{\indiv'\}}{\indiv':\exists R.\{\indiv\}}[$\exists{}^{-1}$]\label{rule: exists inverse}
&&
\tableaurule{\indiv:\Not\exists R^{-1}.C\tand\indiv':\exists R.\{\indiv\}}{\indiv':\Not C}[$\Not\exists{}^{-1}$]%
\label{rule: not exists inverse}
\\
&\tableaurule{\indiv:\exists \Not R.\{\indiv'\}}{\indiv:\Not\exists R.\{\indiv'\}}[$\exists\Not$]%
\label{rule: exists not} 
&&
\tableaurule{\indiv:\Not\exists\Not R.C\tand\indiv':\{\indiv'\}}{\indiv:\exists R.\{\indiv'\}\tor\indiv':\Not C}[$\Not\exists\Not$]%
\label{rule: not exists not} 
\\
&\tableaurule{\indiv:\exists \idRole.\{\indiv'\}}{\indiv:\{\indiv'\}}[$\exists\idRole$]\label{rule: exists Id}
&&
\tableaurule{\indiv:\Not\exists \idRole.C}{\indiv:\Not C}[$\Not\exists\idRole$]%
\label{rule: not exists Id}
\end{align*}
\caption{Tableau calculus $T_\ALBOid$ for \ALBOid.}\label{table: T}
\end{figure}

Let $T_\ALBOid$ be the tableau calculus consisting of the rules listed
in Figure~\ref{table: T}.

The first group of rules are the rules for decomposing concept expressions. 
They are in fact the rules for testing satisfiability of concepts
for the description logic \ALCO, that is, \ALC with individuals.
The \eqref{rule: clash}~rule is the closure rule.
The \eqref{rule: not not}~rule removes occurrences of double negation on concepts. 
(The rule is superfluous if double negations are eliminated using
on-the-fly rewrite rules, but is included to 
simplify the completeness proof somewhat.)
The \eqref{rule: or} and \eqref{rule: not or}~rules are
standard rules for handling concept disjunctions.
As usual, and in accordance with the semantics of the existential
restriction operator, for any existentially restricted concept the
\eqref{rule: exists}~rule creates a new individual with this concept and
adds a link to the new individual.
It is the only rule in the calculus that generates new individuals.
An additional side-condition is that the new individual $\indiv'$
is uniquely associated with the premise $\indiv:\exists R.C$.
As usual we assume that the \eqref{rule: exists}~rule is applied only to expressions of the form
$\indiv:\exists R.C$, when $C$ is not a singleton, that is, $C\neq\{\indiv''\}$
for any individual~$\indiv''$.
This condition prevents the application of the
$\eqref{rule: exists}$~rule to expressions of the form $\indiv:\exists R.\{\indiv''\}$
because this would just cause another witness of an
$R$-successor of $\indiv$ to be created when $\indiv''$ is already such a witness.
The \eqref{rule: not exists}~rule is equivalent  to the standard %
rule for universally restricted concept expressions.
The \eqref{rule: sym}, \eqref{rule: mon}, and \eqref{rule: id}~rules
are the equality rules for individuals,
familiar from hybrid logic tableau systems. They can be viewed
as versions of standard rules for first-order equality.
The \eqref{rule: id}~rule is formulated a bit unusually, because
it adds a validity $\indiv:\{\indiv\}$ to the branch. In our
presentation expressions of the form $\indiv:\{\indiv\}$ are used to
represent the individuals occurring on
a branch, which is exploited by the \eqref{rule: not exists not}~rule. 
See also remarks in Section~\ref{section_discussion}.
The \eqref{rule: nsym}~rule is needed 
to ensure that any negated singleton concept
eventually appears as a label in a concept assertion.

The rules in the second group are the rules for decomposing
complex role expressions.
They can be divided into two subgroups:
rules for positive existential role occurrences and rules for
negated existential role occurrences.
These are listed in the left and right columns, respectively.
Due to the presence of the \eqref{rule: exists}~rule, the rules for positive
existential roles are restricted to role assertions, that is, the
concept expressions immediately below the $\exists$ operator are
singleton concepts.

Among the rules for negated existential roles,
the \eqref{rule: not exists inverse}~rule and
the \eqref{rule: not exists not}~rule
are special.
The \eqref{rule: not exists inverse}~rule
allows the backward propagation of concept expressions
along inverted links (ancestor links).
The \eqref{rule: not exists not}~rule is the rule for the sufficiency
operator.
It expands a universally restricted concept in which the role is
negated according to the semantics: 
\[
x \in (\Not\exists\Not R.C)^\I \iff
\forall y \left( (x,y) \in R^\I \vee y \in (\Not C)^\I\right).
\]
That is, $\indiv'$ in the rule is implicitly quantified by a universal quantifier.
The effect of the second premise, $\indiv':\{\indiv'\}$, is to
instantiate~$\indiv'$ with individuals that occur in the branch.
The remaining rules in this subgroup are based on obvious logical
equivalences in \ALBOid.

Tableau rules that do not produce new individuals
are called \emph{type-completing} rules.
In the case of the calculus~$T_\ALBOid$, with the exception of
the \eqref{rule: exists}~rule, all rules are type-completing.

Now, given an input set of concepts $\mathcal{S}$, a tableau derivation is constructed
as follows.
First, preprocessing is performed. This pushes
occurrences of the role inverse operator in every concept in~$\mathcal{S}$
inward toward atomic concepts by exhaustively
applying the following role equivalences from left to right:
\begin{align*}
(R\Or S)^{-1} & =R^{-1}\Or S^{-1},&
(\Not R)^{-1} & =\Not(R^{-1}),\\
(R^{-1})^{-1} & =R,&
\idRole^{-1} & = \idRole.
\end{align*}
This preprocessing is required to allow the tableau algorithm to
handle the role inverse operator correctly. Pushing inverse inwards is the only
preprocessing required and can be done in linear time.
Transformation to negation normal form is optional
but for practical
purposes not a good choice because obvious properties of expressions and
their negation require unnecessary inference steps to reveal.

Suppose $\mathcal{S}'$ is the set of concepts resulting from
preprocessing.
Then we build a complete tableau $T_\ALBOid(\mathcal{S}')$  
as described above.

In the rest of the paper when we refer to the calculus~$T_\ALBOid$ we assume that
the described preprocessing with respect to role inverse has been
applied to the input concept (set) before the
rules of the calculus are applied.
It is also important to note that $\indiv:C$ and all labelled
expressions and assertions really denote concept expressions.

An example of a finite derivation in $T_\ALBOid$ is shown in Figure~\ref{figure: derivation} for the concept
\[\Not\left(\Not\exists(\rname\Or\Not \rname).\cname\Or\exists \rname.\cname\right).\]
In the figure each line in the derivation is numbered on the left.
The rule applied and the number of the premise(s) to which it was
applied to produce the labelled concept expression (assertion) in each
line is specified on the right.
The black triangles denote branching points in the derivation.
A branch expansion after a branching point is indicated
by appropriate indentation.
The derivation presents a finished tableau with two branches; the
left branch is closed and the right branch is open.
Hence, because the tableau calculus is complete, which is shown in
the next section, the input concept $\Not\left(\Not\exists(\rname\Or\Not \rname).\cname\Or\exists
\rname.\cname\right)$ is satisfiable.

\begin{figure}[!tu]
\begin{center}
\begin{minipage}{.65\textwidth}
  \begin{enumerate}[1.]
    \item\label{ex: drv: a}
        \titem{\indiv_0:\Not(\Not\exists(\rname\Or\Not \rname).\cname\Or\exists \rname.\cname)}{given}
    \item\label{ex: drv: b}
        \titem{\indiv_0:\Not\Not\exists(\rname\Or\Not \rname).\cname)}{\eqref{rule: not or},\ref{ex: drv: a}}
    \item\label{ex: drv: c}
        \titem{\indiv_0:\Not\exists \rname.\cname}{\eqref{rule: not or},\ref{ex: drv: a}}
    \item\label{ex: drv: 0}
        \titem{\indiv_0:\exists(\rname\Or\Not \rname).\cname}{\eqref{rule: not not},\ref{ex: drv: b}}
    \item\label{ex: drv: 1}
        \titem{\indiv_0:\exists(\rname\Or\Not \rname).\{\indiv_1\}}{\eqref{rule: exists},\ref{ex: drv: 0}}
    \item\label{ex: drv: 2}
        \titem{\indiv_1:\cname}{\eqref{rule: exists},\ref{ex: drv: 0}}
    \item\label{ex: drv: 3}
        \titem{\indiv_0:\{\indiv_0\}}{\eqref{rule: id},\ref{ex: drv: a}}
    \item\label{ex: drv: 4}
        \titem{\indiv_1:\{\indiv_1\}}{\eqref{rule: id},\ref{ex: drv: 2}}
    \item\label{ex: drv: 5}
        \titem{\tbranch\indiv_0:\exists \rname.\{\indiv_1\}}{\eqref{rule: exists or},\ref{ex: drv: 1}}
    \item\label{ex: drv: 5.0}
        \titem{\tskip\indiv_1:\Not \cname}{\eqref{rule: not exists},\ref{ex: drv: c}}
    \item\label{ex: drv: 5.1}
        \titem{\tskip\unsat}{\eqref{rule: clash},\ref{ex: drv: 2},\ref{ex: drv: 5.0}}
    \item\label{ex: drv: 6}
        \titem{\tbranch\indiv_0:\exists \Not \rname.\{\indiv_1\}}{\eqref{rule: exists or},\ref{ex: drv: 1}}
    \item\label{ex: drv: 6.0}
        \titem{\tskip\indiv_0:\Not \exists \rname.\{\indiv_1\}}{\eqref{rule: exists or},\ref{ex: drv: 1}}
    \item\titem{\tskip\sat}{(no more rules are applicable)}
  \end{enumerate}
\end{minipage}
\end{center}
\caption{A derivation in $T_\ALBOid$}\label{figure: derivation}
\end{figure}

\section{Soundness and Completeness}
\label{section: Soundness and completeness}

We turn to proving soundness and completeness of the calculus.
It is easy to see that every rule preserves the satisfiability of concept assertions.
More precisely, given an \ALBOid-model, for any rule $\inlinetableaurule{X_0}{X_1\tor\cdots\tor X_m}$
of the calculus $T_\ALBOid$ 
if the set $X_0\sigma$ of instantiations of premises of the rule  
under a substitution $\sigma$
are all satisfiable in the model then 
one of the sets $X_1\sigma,\ldots, X_m\sigma$ is also satisfiable in the model.  
Any rule with this property is said to be \emph{sound}.
Since all the rules of the calculus are sound this implies the calculus
$T_\ALBOid$ is sound.
\begin{theorem}[Soundness]
$T_\ALBOid$ is a sound tableau calculus for satisfiability in
\ALBOid.
\end{theorem}

For proving completeness, suppose that a tableau $T_\ALBOid(\mathcal{S})$ for
a given set of concepts $\mathcal{S}$
is open, that is, 
it contains an open branch~$\branch{B}$.
From $\branch{B}$ we construct a model~$\IB$ for the satisfiability of $\mathcal{S}$
as follows.
By definition, let
\begin{align*}
\indiv\simB\indiv'& \defiff\indiv:\{\indiv'\}\in\branch{B}.
\intertext{%
The rules~\eqref{rule: sym}, \eqref{rule: mon}, and~\eqref{rule: id}
ensure that $\simB$ is an equivalence relation on individuals.
Define the equivalence class $\|\indiv\|$ of an individual~$\indiv$ 
by:
}
\|\indiv\| & \define\{\indiv'\mid\indiv\simB\indiv'\}.
\intertext{%
We set 
$\Delta^\IB \define\{\|\indiv\|\mid\indiv:\{\indiv\}\in\branch{B}\}$,
and for every $\rname\in\Roles$, $\cname\in\Concepts$ and $\indiv\in\Objects$
we define
}
(\|\indiv\|,\|\indiv'\|)\in \rname^\IB&\defiff%
\indiv:\exists \rname.\{\indiv''\}\in\branch{B}\ \text{for some}\ \indiv''\simB\indiv',\\
 \|\indiv\|\in \cname^\IB &\defiff \indiv:\cname\in\branch{B},\\
 \indiv^\IB&\ \,\define\begin{cases}
                \|\indiv\|,& \text{if}\ \indiv:\{\indiv\}\in\branch{B},\\
                \|\indiv'\|& \text{for some}\ \|\indiv'\|\in\Delta^\IB,\ \text{otherwise.}
               \end{cases} 
\end{align*}
The rule~\eqref{rule: mon}
ensures that
the definitions of $\rname^\IB$ and $\cname^\IB$ are correct and 
do not depend on the choice of the representative $\indiv$ of the equivalence
class $\ecl{\indiv}$. 
Finally, we expand the interpretation $\cdot^\IB$ to all concepts and roles in the expected way
using induction on their lengths.
This completes the definition of the model~$\IB$ extracted from a branch $\branch{B}$.

Let $\prec$ be the ordering on expressions (concepts and roles) of \ALBOid induced by the rules of $T_\ALBOid$.
That is, $\prec$ is the smallest transitive ordering
on the set of all \ALBOid expressions 
satisfying:
\begin{align*}
 C & \prec \Not C & R &\prec \Not R&
 R & \prec R^{-1}\\
 C & \prec C\Or D & D & \prec C\Or D\\
 \Not C & \prec \Not(C\Or D) & \Not D &\prec \Not(C\Or D)\\
 C &\prec \exists R.C & R & \prec \exists R.C\\
 \Not C &\prec \Not\exists R.C & R & \prec \Not\exists R.C\\
 R & \prec R\Or S & S & \prec R\Or S
\end{align*}
Note that $\prec$ does not coincide with the direct subexpression ordering. %
$\prec$ is a well-founded ordering, 
since the length of the expression to the left in a clause is always strictly less than the 
length of the expression to the right.
\begin{lemma}\label{lemma: reflection}
\begin{enumerate}[(1)]
\item\label{prop: C}
 If $\indiv:D\in\branch{B}$ then $\|\indiv\|\in D^\IB$ for any concept $D$. 
\item\label{prop: R} For every role $R$ 
                     and every concept $D$
    \begin{enumerate}[(\ref{prop: R}a)]
        \item\label{prop: R: 1} $\indiv:\exists R.\{\indiv'\}\in\branch{B}$ implies $(\|\indiv\|,\|\indiv'\|)\in R^\IB$,
        \item\label{prop: R: 2} if $(\|\indiv\|,\|\indiv'\|)\in R^\IB$ and $\indiv:\Not\exists R.D\in\branch{B}$
              then $\indiv':\Not D\in\branch{B}$.
    \end{enumerate}
\end{enumerate}
\end{lemma}
\begin{proof}
We prove both properties simultaneously by induction on the 
ordering $\prec$.
The induction hypothesis is:
for an arbitrary \ALBOid expression $E$
and each expression~$F$ such that $F\prec E$,
if $F$ is a concept then property~\eqref{prop: C} holds with $D=F$.
Otherwise (that is, if $F$ is a role), property~\eqref{prop: R} holds with $R=F$ (and $D$ is arbitrary).

To prove property~\eqref{prop: C} we consider the following cases.
\begin{description}
 \item[$D=\cname$] We have $\indiv: \cname\in\branch{B}\iff\|\indiv\|\in \cname^\IB$ by definition of $\cname^\IB$.
 \item[$D=\{\indiv'\}$] If $\indiv:\{\indiv'\}\in\branch{B}$ then, using the definition of $\indiv^\IB$,
      $\indiv^\IB=\|\indiv\|=\|\indiv'\|=\indiv'^\IB$ and,
      consequently, $\indiv^\IB\in\{\indiv'^\IB\}=\{\indiv'\}^\IB$.
 \item[$D=\Not D_0$] If $\indiv:\Not D_0\in\branch{B}$ then $\indiv:D_0\notin\branch{B}$ because
    otherwise $\branch{B}$ would have been closed by the \eqref{rule: clash}~rule.
    We have the following subcases.
    \begin{description}
     \item[$D_0=\cname$] We have $\indiv: \cname\in\branch{B}\iff\|\indiv\|\in \cname^\IB$ by definition of $\cname^\IB$.
     \item[$D_0=\{\indiv'\}$] As the rules~\eqref{rule: nsym} and~\eqref{rule: id} have been applied in $\branch{B}$
                             it is clear
                             that $\indiv':\{\indiv'\}$ is in $\branch{B}$ and $\indiv'^\IB=\|\indiv'\|$.
                             Similarly, $\indiv^\IB=\|\indiv\|$.
                             Furthermore, because $\indiv:\{\indiv'\}\notin\branch{B}$ we have $\indiv\notsimB\indiv'$,
                             that is $\indiv\notin\|\indiv'\|=\indiv'^\IB$. 
                             Consequently, $\|\indiv\|\notin\{\|\indiv'\|\}=\{\indiv'\}^\IB$.
     \item[$D_0=\Not D_1$] 
        If $\indiv:\Not\Not D_1\in\branch{B}$ then $\indiv:D_1\in\branch{B}$ 
        by the \eqref{rule: not not}~rule.
        By the induction hypothesis, the property~\eqref{prop: C} holds for $D_1\prec\Not\Not D_1$.
        Thus, $\|\indiv\|\in D_1^\IB=(\Not\Not D_1)^\IB=D^\IB$.
     \item[$D_0=D_1\Or D_2$] In this case both $\indiv:\Not D_1$ and $\indiv:\Not D_2$ are in the branch $\branch{B}$
        by the \eqref{rule: not or}~rule. We also have $\Not D_1\prec\Not(D_1\Or D_2)$ and 
        $\Not D_2\prec\Not(D_1\Or D_2)$.
        Hence, by the induction hypothesis for the property~\eqref{prop: C}, 
        $\|\indiv\|\in(\Not D_1)^\IB\cap(\Not D_2)^\IB=(\Not(D_1\Or D_2))^\IB$.
     \item[$D_0=\exists R.D_1$] Let $\|\indiv'\|$ be an arbitrary element of $\Delta^\IB$ such that
        $(\|\indiv\|,\|\indiv'\|)\in R^\IB$ (if there is no such element then there is nothing to prove).
        By the induction hypothesis the property~\eqref{prop: R: 2} holds
        for $R\prec D_0$. Thus, $\indiv':\Not D_1\in\branch{B}$. The induction hypothesis
        for the property~\eqref{prop: C} gives us $\|\indiv'\|\notin D_1^\IB$. Finally,
        we obtain $\|\indiv\|\in(\Not\exists R.D_1)^\IB$ because $\indiv'$ was chosen arbitrarily.
    \end{description}

 \item[$D=D_0\Or D_1$]
    If $\indiv:D_0\Or D_1\in\branch{B}$ then either $\indiv:D_0\in\branch{B}$, or $\indiv:D_1\in\branch{B}$
    by the \eqref{rule: or}~rule.
    Hence, either $\|\indiv\|\in D_0^\IB$ or $\|\indiv\|\in D_1^\IB$ by the induction hypothesis for
    $D_0\prec D$ and $D_1\prec D$.
    Thus, $\|\indiv\|\in D_0^\IB\cup D_1^\IB=D^\IB$.
 \item[$D=\exists R.D_0$]
    If $\indiv:\exists R.D_0\in\branch{B}$ then $\indiv':D_0\in\branch{B}$ and $\indiv:\exists R.\{\indiv'\}\in\branch{B}$ 
    for some individual $\indiv'$ by the \eqref{rule: exists}~rule. 
    By the induction hypothesis the properties~\eqref{prop: C} and~\eqref{prop: R: 1} hold for $D_0\prec D$
    and $R\prec D$ respectively. 
    Hence,
    $\|\indiv'\|\in D_0^\IB$ and $(\|\indiv\|,\|\indiv'\|)\in R^\IB$. That is, $\|\indiv\|\in(\exists R.D_0)^\IB$.
\end{description}
To prove property~\eqref{prop: R} we consider all cases corresponding to the possible forms that a role $R$ can have.
\begin{description}
 \item[$R=\rname$]
    Since $\indiv:\exists \rname.\{\indiv'\}\in\branch{B}$ implies $(\|\indiv\|,\|\indiv'\|)\in \rname^\IB$
    using the definition of $\rname^\IB$, property~\eqref{prop: R: 1} holds trivially.
    For~\eqref{prop: R: 2} let $(\|\indiv\|,\|\indiv'\|)\in \rname^\IB$ and $\indiv:\Not\exists \rname.D\in\branch{B}$.
    Using the definition of $\rname^\IB$, there is an $\indiv''\simB\indiv'$ such that $\indiv:\exists \rname.\{\indiv''\}\in\branch{B}$.
    Therefore, by the \eqref{rule: not exists}~rule, $\indiv'':\Not D$ is in the branch $\branch{B}$, too.
    Finally, $\indiv':\Not D\in\branch{B}$ by the \eqref{rule: mon}~rule.
 \item[$R=\idRole$] For property~\eqref{prop: R: 1} let $\indiv:\exists \idRole.\{\indiv'\}\in\branch{B}$.
                    Then $\indiv:\{\indiv'\}\in\branch{B}$ by the \eqref{rule: exists Id}~rule and, hence, $\ecl{\indiv}=\ecl{\indiv'}$.
                    Consequently, $(\|\indiv\|,\|\indiv'\|)\in\idRole^\IB$.
                    For~\eqref{prop: R: 2} suppose that 
                    $(\|\indiv\|,\|\indiv'\|)\in \idRole^\IB$ and $\indiv:\Not\exists \idRole.D\in\branch{B}$.
                    Hence, $\ecl{\indiv}=\ecl{\indiv'}$ and, therefore, $\indiv:\{\indiv'\}$ is in $\branch{B}$.
                    Further, by the \eqref{rule: not exists Id}~rule,
                     $\indiv:\Not D\in\branch{B}$.
                    Finally, by the \eqref{rule: mon}~rule, 
                    $\indiv':\Not D$ is in the branch $\branch{B}$.
 \item[$R=S^{-1}$] For property~\eqref{prop: R: 1} let $\indiv:\exists S^{-1}.\{\indiv'\}\in\branch{B}$.
                    Then $\indiv':\exists S.\{\indiv\}\in\branch{B}$ by the rule~\eqref{rule: exists inverse}.
                    By the induction hypothesis for $S\prec R$ we have $(\|\indiv'\|,\|\indiv\|)\in S^\IB$.
                    Consequently, $(\|\indiv\|,\|\indiv'\|)\in (S^{-1})^\IB$.
                    For~\eqref{prop: R: 2} suppose that 
                    $(\|\indiv\|,\|\indiv'\|)\in (S^{-1})^\IB$ and $\indiv:\Not\exists S^{-1}.D\in\branch{B}$.
                    As all the occurrences of the inverse operator have been pushed through other role connectives
                    and double occurrences
                    of ${}^{-1}$ have been removed\footnote{%
                        Removing double occurrences of the inverse operator in front of atoms
                        is not essential for the proof but it simplifies matters a bit.}
                    we can assume that $S=\rname$ for some role name $\rname$.
                    Hence, $(\|\indiv'\|,\|\indiv\|)\in \rname^\IB$ and, by the definition of $\rname^\IB$,
                    there is $\indiv''\simB\indiv$ such that $\indiv':\exists \rname.\{\indiv''\}\in\branch{B}$.
                    By the \eqref{rule: mon}~rule, $\indiv'':\Not\exists \rname^{-1}.D\in\branch{B}$.
                    Finally, by the \eqref{rule: not exists inverse}~rule, 
                    $\indiv':\Not D$ is in the branch $\branch{B}$.
 \item[$R=S_0\Or S_1$] 
    For~\eqref{prop: R: 1} suppose $\indiv:\exists (S_0\Or S_1).\{\indiv'\}\in\branch{B}$.
    Hence, $\indiv:\exists S_0.\{\indiv'\}\in\branch{B}$ or $\indiv:\exists S_1.\{\indiv'\}\in\branch{B}$
    by the \eqref{rule: exists or}~rule.
    Thus, by the induction hypothesis for $S_0\prec R$ and $S_1\prec R$ we have
    either $(\|\indiv\|,\|\indiv'\|)\in S_0^\IB$ or $(\|\indiv\|,\|\indiv'\|)\in S_1^\IB$.
    Finally, by the semantics of the relational $\Or$ connective we obtain
    $(\|\indiv\|,\|\indiv'\|)\in (S_0\Or S_1)^\IB$.
    For~\eqref{prop: R: 2} let $(\|\indiv\|,\|\indiv'\|)\in (S_0\Or S_1)^\IB=S_0^\IB\cup S_1^\IB$
        and $\indiv:\Not\exists(S_0\Or S_1).D\in\branch{B}$.
    By the \eqref{rule: not exists or}~rule we obtain that both
        $\indiv:\Not\exists S_0.D$ and $\indiv:\Not\exists S_1.D$ are in $\branch{B}$.
    Therefore, by the induction hypothesis for $S_0\prec R$ and $S_1\prec R$ the property~\eqref{prop: R: 2} holds
    and, thus, we have $\indiv':\Not D\in\branch{B}$.
 \item[$R=\Not S$] For~\eqref{prop: R: 1} suppose $\indiv:\exists\Not S.\{\indiv'\}\in\branch{B}$. 
    Then $\indiv:\Not\exists S.\{\indiv'\}\in\branch{B}$ is obtained with the \eqref{rule: exists not}~rule. 
    If $(\|\indiv\|,\|\indiv'\|)\notin(\Not S)^\IB$
    then $(\|\indiv\|,\|\indiv'\|)\in S^\IB$ and by property~\eqref{prop: R: 2} which holds
    by the induction hypothesis for $S\prec R$ we have that
    $\indiv':\Not\{\indiv'\}$ is in $\branch{B}$. This concept together with $\indiv':\{\indiv'\}$ implies
    the branch is closed. We reach a contradiction, so $(\|\indiv\|,\|\indiv'\|)\in(\Not S)^\IB$.

    For property~\eqref{prop: R: 2} suppose that $(\|\indiv\|,\|\indiv'\|)\in(\Not S)^\IB$
    and $\indiv:\Not\exists\Not S.D$ are in the branch $\branch{B}$.
    Then we have $(\|\indiv\|,\|\indiv'\|)\notin  S^\IB$ and, hence,
    by the contrapositive of property~\eqref{prop: R: 1} for $S\prec R$,
    $\indiv:\exists S.\{\indiv'\}$ is not in $\branch{B}$.
    Applying the 
    \eqref{rule: not exists not}~rule to $\indiv:\Not\exists\Not S.D$ we
get
    $\indiv':\Not D\in\branch{B}$.
\end{description}
\end{proof}
A direct consequence of this lemma is a stronger form of completeness
of the tableau calculus which states existence of a model for an open
branch.
From this it follows that if the input set is unsatisfiable a closed tableau
can be constructed for it.
\begin{theorem}
[Completeness]\label{theorem: completeness} For any set of concepts
$\mathcal{S}$ and any tableau $T_\ALBOid(\mathcal{S})$, if there
is an open branch in $T_\ALBOid(\mathcal{S})$ then $\mathcal{S}$
is satisfiable in an \ALBOid-model.
This implies
$T_\ALBOid$ is a complete tableau calculus for testing satisfiability
of concept expressions in \ALBOid.
\end{theorem}
Thus, the calculus $T_\ALBOid$ is sound and complete for reasoning in \ALBOid.

\section{Unrestricted Blocking}
\label{section_blocking}

There are satisfiable concepts which can result in an infinite
$T_\ALBOid$-tableau, where all open branches are infinite.
The concept \[\Not\exists(Q'\Or\Not Q').\Not\exists \rname.\cname\] is such an
example. 
Since the prefix 
$\Not\exists(Q'\Or\Not Q').\Not$ is equivalent to the universal modality,
the concept $\indiv:\exists \rname.\cname$ is propagated to every
individual $\indiv$ in every branch of the tableau.
The concept $\indiv:\exists \rname.\cname$ itself, each time triggers the creation
of a new individual with the \eqref{rule: exists}~rule.
Thus, any branch of the tableau contains infinitely
many individuals. The branches
have however a regular structure that can be detected with loop
detection or blocking mechanisms.

\begin{figure}[!tu]
\begin{center}
\begin{minipage}{.65\textwidth}
  \begin{enumerate}[1.]
   \item\label{ex: std.l.c.:0}\titem{\indiv_0:\Box\exists \rname.\cname}{given}
   \item\label{ex: std.l.c.:0b}\titem{\indiv_0:\{\indiv_0\}}{\eqref{rule: id},\ref{ex: std.l.c.:0}}
   \item\label{ex: std.l.c.:1}\titem{\indiv_0:\exists \rname.\cname}{\eqref{rule: universal modality},\ref{ex: std.l.c.:0},\ref{ex: std.l.c.:0b}}
   \item\label{ex: std.l.c.:2}\titem{\indiv_1:\cname}{\eqref{rule: exists},\ref{ex: std.l.c.:1}}
   \item\label{ex: std.l.c.:3}\titem{\indiv_0:\exists \rname.\{\indiv_1\}}{\eqref{rule: exists},\ref{ex: std.l.c.:1}}
   \item\label{ex: std.l.c.:2a}\titem{\indiv_1:\{\indiv_1\}}{\eqref{rule: id},\ref{ex: std.l.c.:2}}
   \item\label{ex: std.l.c.:4}\titem{\indiv_1:\exists \rname.\cname}{\eqref{rule: universal modality},\ref{ex: std.l.c.:0},\ref{ex: std.l.c.:2a}}
   \item\label{ex: std.l.c.:4a}\titem{\indiv_0\not\sim\indiv_1}%
    {Loop checking: $\cname\in\mathcal{L}(\indiv_1)\setminus\mathcal{L}(\indiv_0)$}
   \item\label{ex: std.l.c.:5}\titem{\indiv_2:\cname}{\eqref{rule: exists},\ref{ex: std.l.c.:4}}
   \item\label{ex: std.l.c.:6}\titem{\indiv_1:\exists \rname.\{\indiv_2\}}{\eqref{rule: exists},\ref{ex: std.l.c.:4}}
   \item\label{ex: std.l.c.:5a}\titem{\indiv_2:\{\indiv_2\}}{\eqref{rule: id},\ref{ex: std.l.c.:5}}
   \item\label{ex: std.l.c.:7}\titem{\indiv_2:\exists \rname.\cname}{\eqref{rule: universal modality},\ref{ex: std.l.c.:0},\ref{ex: std.l.c.:5a}}
   \item\label{ex: std.l.c.:8}\titem{\indiv_1\sim\indiv_2}%
    {Loop checking: $\mathcal{L}(\indiv_1)\supseteq\mathcal{L}(\indiv_2)$}
  \end{enumerate}
\end{minipage}
\end{center}
\caption{Standard loop checking mechanism in $\ALCO^\Box$.}\label{figure: standard loop checking}
\end{figure}

Let us demonstrate how \emph{standard loop checking} (in this case,
subset ancestor blocking, see for example~\citeN{BaaderCalvaneseEtal03}) detects a loop.
As we already observed, satisfiability of the concept $\Not\exists(Q'\Or\Not
Q').\Not\exists \rname.\cname$ corresponds to satisfiability of the 
concept $\Box \exists \rname.\cname$
in the description logic $\ALCO^\Box$, which is \ALCO with the universal modality~$\Box$.
As a calculus for $\ALCO^\Box$ 
we use the first group of rules from Figure~\ref{table: T} plus 
the following two rules (where $\indiv'$ in the left rule is uniquely
associated with $\indiv:\Not\Box C$).
\begin{align*}
 \tableaurule{\indiv:\Not\Box C}{\indiv':\Not C}[$\Not\Box$]\label{rule: existential modality}\ \text{($\indiv'$ is new)}
&&
 \tableaurule{\indiv:\Box C\tand \indiv':\{\indiv'\}}{\indiv':C}[$\Box$]\label{rule: universal modality}
\end{align*}
Similar to the proofs in Section~\ref{section: Soundness and completeness}
it can be proved that
this calculus is sound and complete for satisfiability problem in $\ALCO^\Box$.

Figure~\ref{figure: standard loop checking} gives a derivation in this
calculus using loop checking based on standard subset ancestor blocking.
For any individual $\indiv$,
let $\mathcal{L}(\indiv)$ be a set of concepts associated with $\indiv$
in the current branch. More precisely,
it is defined by
\[
 \mathcal{L}(\indiv)\define\{C\mid\indiv:C\ \text{is in the current branch, and $C\neq\{\indiv\}$}\}.%
\]
After all the type-completing rules have been applied to all concept
expressions labelled with a specific individual, loop checking tests
are performed relative to an ancestor individual.
Two loop checking tests are performed, namely in step~\ref{ex: std.l.c.:4a}
and step~\ref{ex: std.l.c.:8}.
Consider step~\ref{ex: std.l.c.:8}.
All the type-completing rules have been applied to all concept expressions of the
form $\indiv_2:C$ and $\indiv_1:C$.
Comparison of the sets of concepts $\mathcal{L}(\indiv_1)$ and  $\mathcal{L}(\indiv_2)$ 
associated with $\indiv_1$ and $\indiv_2$ in the loop checking test shows that 
$\mathcal{L}(\indiv_1)$ subsumes $\mathcal{L}(\indiv_2)$.
Thus they are in a subset relationship as indicated, and
consequently the individuals $\indiv_1$ and $\indiv_2$
can be identified.
At step~\ref{ex: std.l.c.:4a} the sets 
$\mathcal{L}(\indiv_0)$ and $\mathcal{L}(\indiv_1)$
are not in a subset
relationship
because $\indiv_0:\cname$ is not present in the branch.
The derivation therefore cannot yet stop, but does in step~\ref{ex: std.l.c.:8}.

This example illustrates one of the simplest forms of standard loop checking
used in description and modal logic tableau procedures.
Other forms of loop-checking have been devised for different logics.
One can classify the different existing loop-checking mechanisms as using a
combination of these blocking techniques:
subset or equality blocking, non-pairwise or pairwise blocking,
ancestor or anywhere blocking, and static or dynamic blocking, see for example~\cite{BaaderSattler-OTADL-2001}.
These techniques are all based on comparing sets of concepts labelled by some
individuals (and in the case of pairwise blocking also sets of roles
associated to predecessors).
It is not clear whether these forms of loop checking are sufficient to handle role
negation though.

Suppose that, at some node of a branch, all the type-completing
rules have been applied to concepts labelled with individuals $\indiv_0$
and $\indiv_1$.
Following the standard non-pairwise blocking procedure
we can identify $\indiv_0$ and $\indiv_1$
if they both label the same set of concepts.
However, it is not correct to start identifying the individuals
at this point in the derivation as there could be, for example, a concept $\exists
S.\Not\exists\Not R.C$ involving role negation, where $R$ and $C$ are expressions that are suitably complex.
At a subsequent point the 
\eqref{rule: exists}~rule is applied to this concept and, hence,
the expression $\indiv_2:\Not\exists\Not R.C$,
where $\indiv_2$ is new,
appears in the branch. This triggers the application of
the \eqref{rule: not exists not} rule. Because of the form of $R$
and $\neg C$ in the two branches it also triggers application of the
\eqref{rule: not exists}~rule and
possibly other type-completing rules.
After such applications, it can happen that the labels are not
identifiable (anymore)
because their types have become distinguishable by some concept. 
That is, any introduction of a new individual in a tableau 
has, in general, a global effect on the provisional model
constructed so far.

The example in Figure~\ref{figure: global effect} illustrates this
global effect.
Here we are interested in the satisfiability of 
\[
   (\exists \rname.\cname)\And
   (\exists Q'.\cname)\And
   (\exists Q''.\forall Q''.\forall Q'^{-1}.\cbot)\And
   (\forall Q''.\overline\forall Q''.\forall Q'^{-1}.\cbot).
\]
(Recall that $\forall Q''.C=\Not\exists Q''.\Not C$ and $\overline\forall Q''.C=\Not\exists\Not Q''.C$.)
At step~\ref{ex: gl.eff.:6a}
none of the type-completing rules need to be applied to
concepts labelled with $\indiv_1$ and $\indiv_2$.
Although at this point $\indiv_1$ and $\indiv_2$ are labels of the same subconcepts of the
given concepts,
we cannot make them equal (using equality anywhere blocking).
The reason is that
in step~\ref{ex: gl.eff.:8}
a new individual is introduced 
which causes a few applications of the \eqref{rule: not exists}~rule,
and
as a result, at step~\ref{ex: gl.eff.:13.0}, the types of
$\indiv_1$ and $\indiv_2$ are now distinguished by the concept $\exists Q'^{-1}.(\cname\Or\Not \cname)$.
The example is designed to show that it is not possible
to use standard blocking techniques
to identify individuals ($\indiv_1$ and $\indiv_2$)
although they are locally complete (that is, none of the type completing rules can be applied to them),
because other individuals ($\indiv_0$) can influence $\indiv_1$ and~$\indiv_2$
via the right branch of the \eqref{rule: not exists not}~rule.

\begin{figure}[tbu]
\begin{center}
\begin{minipage}{.65\textwidth}
  \begin{enumerate}[1.]
   \item\label{ex: gl.eff.:0}\titem{\indiv_0:\exists \rname.\cname}{given}
   \item\label{ex: gl.eff.:1}\titem{\indiv_0:\exists Q'.\cname}{given}
   \item\label{ex: gl.eff.:2}\titem{\indiv_0:\exists Q''.\Not\exists Q''.\exists Q'^{-1}.(\cname\Or\Not \cname)}{given}
   \item\label{ex: gl.eff.:3}\titem{\indiv_0:\Not\exists Q''.\exists\Not Q''.\Not\exists Q'^{-1}.(\cname\Or\Not \cname)}{given}
   \item\label{ex: gl.eff.:4}\titem{\indiv_1:\cname}{\eqref{rule: exists},\ref{ex: gl.eff.:0}}
   \item\label{ex: gl.eff.:5}\titem{\indiv_0:\exists \rname.\{\indiv_1\}}{\eqref{rule: exists},\ref{ex: gl.eff.:0}}
   \item\label{ex: gl.eff.:4a}\titem{\indiv_1:\{\indiv_1\}}{\eqref{rule: id},\ref{ex: gl.eff.:4}}
   \item\label{ex: gl.eff.:6}\titem{\indiv_2:\cname}{\eqref{rule: exists},\ref{ex: gl.eff.:1}}
   \item\label{ex: gl.eff.:7}\titem{\indiv_0:\exists Q'.\{\indiv_2\}}{\eqref{rule: exists},\ref{ex: gl.eff.:1}}
   \item\label{ex: gl.eff.:6a}\titem{\indiv_2:\{\indiv_2\}}{\eqref{rule: id},\ref{ex: gl.eff.:6}}
   \item\label{ex: gl.eff.:8}\titem{\indiv_3:\Not\exists Q''.\exists Q'^{-1}.(\cname\Or\Not \cname)}{\eqref{rule: exists},\ref{ex: gl.eff.:2}}
   \item\label{ex: gl.eff.:9}\titem{\indiv_0:\exists Q''.\{\indiv_3\}}{\eqref{rule: exists},\ref{ex: gl.eff.:2}}
   \item\label{ex: gl.eff.:10}\titem{\indiv_3:\Not\exists\Not Q''.\Not\exists Q'^{-1}.(\cname\Or\Not \cname)}%
    {\eqref{rule: not exists},\ref{ex: gl.eff.:3},\ref{ex: gl.eff.:9}}
   \item\label{ex: gl.eff.:11}\titem{\tbranch\indiv_3:\exists Q''.\{\indiv_2\}}%
    {\eqref{rule: not exists not},\ref{ex: gl.eff.:10},\ref{ex: gl.eff.:6a}}

   \item\label{ex: gl.eff.:11.0}\titem{\tskip\indiv_2:\Not\exists Q'^{-1}.(\cname\Or\Not \cname)}%
    {\eqref{rule: not exists},\ref{ex: gl.eff.:8},\ref{ex: gl.eff.:11}}
   \item\label{ex: gl.eff.:11.1}\titem{\tskip\indiv_0:\Not(\cname\Or\Not \cname)}%
    {\eqref{rule: not exists inverse},\ref{ex: gl.eff.:11.0},\ref{ex: gl.eff.:7}}
   \item\label{ex: gl.eff.:11.2}\titem{\tskip\unsat}{after a few steps}
   \item\label{ex: gl.eff.:12.0}\titem{\tbranch\indiv_2:\Not\Not\exists Q'^{-1}.(\cname\Or\Not \cname)}%
    {\eqref{rule: not exists not},\ref{ex: gl.eff.:10},\ref{ex: gl.eff.:6a}}
   \item\label{ex: gl.eff.:12.1}\titem{\tskip\indiv_2:\exists Q'^{-1}.(\cname\Or\Not \cname)}%
    {\eqref{rule: not not},\ref{ex: gl.eff.:12.0}}
   \item\label{ex: gl.eff.:13}\titem{\tskip\tbranch\indiv_3:\exists Q''.\{\indiv_1\}}%
    {\eqref{rule: not exists not},\ref{ex: gl.eff.:10},\ref{ex: gl.eff.:4a}}
   \item\label{ex: gl.eff.:13.0}\titem{\tskip[2]\indiv_1:\Not\exists Q'^{-1}.(\cname\Or\Not \cname)}%
    {\eqref{rule: not exists},\ref{ex: gl.eff.:8},\ref{ex: gl.eff.:13}}
   \item[\ldots]\titem{\tskip[2]{}}{}
   \item\label{ex: gl.eff.:14}\titem{\tskip\tbranch\indiv_1:\Not\Not\exists Q'^{-1}.(\cname\Or\Not \cname)}%
    {\eqref{rule: not exists not},\ref{ex: gl.eff.:10},\ref{ex: gl.eff.:4a}}

   \item[\ldots]\titem{\tskip[2]{}}{}
  \end{enumerate}
\end{minipage}
\end{center}
 \caption{Global effect of the introduction of a new individual}\label{figure: global effect}
\end{figure}

The examples illustrate that a reason for non-termination of
$T_\ALBOid$ is the possible infinite generation of labels.
Although the rules respect the well-founded ordering~$\prec$
they do it only with respect to expressions, not labelled
expressions (assertions).
Furthermore, it is easy to see that only applications of the \eqref{rule: exists}~rule generate
new individuals in the branch. Thus, 
a reason that a branch can be infinite 
is the unlimited application of the \eqref{rule: exists}~rule; it is in
fact the only possible reason.
It is crucial therefore to have a blocking mechanism that avoids
infinite derivations by restricting the application
of the \eqref{rule: exists}~rule.

In order to turn the calculus $T_\ALBOid$ into a terminating calculus
for \ALBOid, we introduce a new approach to blocking.

Let $<$ be an ordering on individuals in the branch which is a linear
extension of the order of introduction of the individuals during
the derivation.
That is, let $\indiv<\indiv'$, whenever
the first appearance of individual $\indiv'$ in
the branch is strictly later than the first appearance of individual
$\indiv$.

We add the following rule, called the \emph{unrestricted blocking}
rule, to the calculus. 
\[
\tableaurule{\indiv:\{\indiv\}\tand\indiv':\{\indiv'\}}{\indiv:\{\indiv'\}\tor\indiv:\Not\{\indiv'\}}[ub]
\label{rule: unrestricted blocking}
\]
Moreover, we require that the following conditions both hold.
\begin{enumerate}[(c1)]
\item\label{req: blocking}
If $\indiv:\{\indiv'\}$ appears in a branch and $\indiv<\indiv'$ then
possible applications of the \eqref{rule: exists}~rule to expressions
of the form $\indiv':\exists R.C$
are not performed within the branch.
\item\label{req: ub before exists}
In every open branch there is some node from which point onwards
before any application of the \eqref{rule: exists}~rule
all possible applications of the \eqref{rule: unrestricted blocking}~rule
have been performed. 
\end{enumerate}

The intuition of the \eqref{rule: unrestricted blocking}~rule is that
it conjectures whether two labels are equal or not.
In the left branch two labels are set to be equal. If this does not 
lead to finding a model then the labels cannot be equal. This is the
information carried by the right branch.
Conditions (c\ref{req: blocking}) and (c\ref{req: ub before exists})
are important to restrict the application of the \eqref{rule: exists}~rule.
Condition~(c\ref{req: blocking}) specifies that it may only be applied to labelled expressions
where the label is the smallest representative of an equivalence class.
Condition~(c\ref{req: ub before exists}) says that from some point
onwards in a branch blocking has been applied exhaustively before
the application of the \eqref{rule: exists}~rule.

We use the notation \TALBOidub for the extension of $T_\ALBOid$ with 
this blocking mechanism using the \eqref{rule: unrestricted blocking}~rule.

\begin{theorem}
 \TALBOidub is a sound and complete tableau calculus for \ALBOid.
\end{theorem}
\begin{proof}
The \eqref{rule: unrestricted blocking}~rule is sound in the usual
sense (see remarks at the beginning of Section~\ref{section: Soundness and completeness}).
The blocking conditions (c\ref{req: blocking}) and (c\ref{req: ub before exists}) are sound in the sense that
they cannot cause an open branch to become closed.
Hence, the blocking rule and the blocking conditions can be safely
added to \emph{any} tableau calculus without endangering soundness
or completeness.
As $T_\ALBOid$ is sound and complete, it follows that $\TALBOidub$
is sound and complete.
\end{proof}

\section{Termination Through Unrestricted Blocking}
\label{section_termination}

Next we prove termination.
Since every finite set of concepts can be replaced by the conjunction
of its elements, without loss of generality in this section and the
next section, we restrict ourselves to a single input concept.

For every set $X$, let $\Card(X)$ denote the \emph{cardinality} of $X$. 
Let $\preceq$ be reflexive closure of the ordering $\prec$ which has been defined in Section~\ref{section: Soundness and completeness}.
For every concept $D$ we define $\sub(D)\define\{D'\mid D'\preceq D\}$.

Let $C$ be the given input concept and $n$ be the length of the concept $C$. %
Let $k$ be the number of individuals occurring in the input concept assertion $\indiv:C$.
Suppose $\branch{B}$ is an arbitrary open branch in a $\TALBOidub$ tableau
for $C$ and $\IB$ be the model constructed
from~$\branch{B}$ as described in Section~\ref{section: Soundness and completeness}.

For every $\|\indiv\|\in\Delta^\IB$, let $\napplE(\|\indiv\|)$
denote the number of applications in $\branch{B}$ of the \eqref{rule: exists}~rule
to concepts of the form $\indiv':\exists R.D$
with $\indiv'\in\|\indiv\|$.

\begin{lemma}\label{lemma: napplE(ell)}
$\napplE(\|\indiv\|)$ is finite for every $\|\indiv\|\in\Delta^\IB$.
\end{lemma}
\begin{proof}
Suppose not, that is, $\napplE(\|\indiv\|)$ is infinite.
The number of concepts of the form $\exists R.D$
where $D$ is not a singleton under labels in the branch is finitely bounded.
Because the \eqref{rule: exists}~rule
is not allowed to be applied to
role assertions (that is, concepts of the form $\indiv:\exists R.\{\indiv'\}$)
there is a sequence of individuals
$\indiv_0,\indiv_1,\ldots$ such that every
$\indiv_i\in\|\indiv\|$ and
the \eqref{rule: exists}~rule has been applied
to concepts $\indiv_0:\exists R.D, \indiv_1:\exists R.D,\ldots$
for some $\exists R.D\preceq C$.
However, such a situation is impossible because of
conditions~(c\ref{req: ub before exists}) and~(c\ref{req: blocking}).
Indeed, without loss of generality we can assume that $\indiv<\indiv_0<\indiv_1<\cdots$.
Then, 
by Condition~(c\ref{req: ub before exists}),
starting from some node of $\branch{B}$,
as soon as $\indiv_i$ appears in $\branch{B}$,
it is detected that $\indiv_i\in\|\indiv\|$
before any next
application of the \eqref{rule: exists}~rule
and, hence, $\indiv_i$ is immediately blocked for any application
of the rule~\eqref{rule: exists}, due to Condition~(c\ref{req: blocking}).
\end{proof}

Let $\napplE(\branch{B})$ denote the number of applications
of the \eqref{rule: exists}~rule
in a branch $\branch{B}$.
Since the \eqref{rule: exists}~rule is the
only rule that generates new individuals, 
$\napplE(\branch{B})$ is in fact the number of generated
individuals in the branch $\branch{B}$
and the number of all individuals occurring in~$\branch{B}$ 
is $k+\napplE(\branch{B})$.
For the the number of all individuals occurring in~$\branch{B}$
we introduce the notation 
$i(\branch{B})$, that is, $i(\branch{B})\define k+\napplE(\branch{B})$.

Let $M(\branch{B})$ be the least cardinal 
which is the upper bound for all $\napplE(\|\indiv\|)$ 
where $\|\indiv\|$ ranges all the equivalence classes from $\Delta^\IB$.
That is:
\[M(\branch{B})\define\sup\{\napplE(\|\indiv\|)\mid\|\indiv\|\in\Delta^\IB\}.\]
The following lemma which establishes the upper bound for $\napplE(\branch{B})$ is easy to prove.

\begin{lemma}\label{lemma: finite number of nominals}
$\napplE(\branch{B})\leq M(\branch{B})\cdot\Card(\Delta^\IB)$.
\end{lemma}
Furthermore, the following lemma holds.
\begin{lemma}\label{lemma: concept number in branch}
The number of concepts in $\branch{B}$ does not exceed $n\cdot i(\branch{B})+2n\cdot i(\branch{B})^2$.
In particular, if $\napplE(\branch{B})$ is finite then $\branch{B}$ is finite.
\end{lemma}
\begin{proof}
Let $D\in\branch{B}$. Then $D$
has one of the following forms:
\begin{align*}
D&=\indiv:D',\quad \text{where $D'$ is a concept in $\sub(C)$,}\\
D&=\indiv:\{\indiv'\},\\
D&=\indiv:\Not\{\indiv'\},\\
D&=\indiv:\exists R.\{\indiv'\},\quad \text{where $R$ is a role in $\sub(C)$,}\\
D&=\indiv:\Not\exists R.\{\indiv'\},\quad \text{where $R$ is a role in $\sub(C)$.}
\end{align*}
The number of individuals and concepts in $\sub(C)$ is bounded by $n$
and the number of roles in $\sub(C)$ is strictly less than $n$.
Hence, an upper bound for the number of concepts in a branch~$\branch{B}$
can be given as:
$n\cdot i(\branch{B}) + i(\branch{B}) ^2+
 i(\branch{B}) ^2+(n-1)\cdot i(\branch{B}) ^2+(n-1)\cdot i(\branch{B}) ^2
= n\cdot i(\branch{B}) +2n\cdot i(\branch{B}) ^2.$
\end{proof}
Clearly if $\Delta^\IB$ is finite then 
$M(\branch{B})=\max\{\napplE(\|\indiv\|)\mid\|\indiv\|\in\Delta^\IB\}$.
Therefore, as a consequence of Lemmas~\ref{lemma: finite number of nominals} and~\ref{lemma: concept number in branch},
we obtain the following statement.
\begin{corollary}\label{corollary: branch finiteness criterion}
Let $\branch{B}$ be an open branch in a $\TALBOidub$ tableau for a concept $C$.
Then the branch~$\branch{B}$ is finite iff $\Delta^\IB$ is finite.
\end{corollary}

\begin{lemma}\label{lemma: some magic}
Assume that the input concept $C$ is satisfiable in a model $\I$.
Then there is an open branch $\branch{B}$ in $\TALBOidub(C)$ such that
$\Card(\Delta^\IB)\leq\Card(\Delta^\I)$.
\end{lemma}
\begin{proof}
By induction on application of the rules in $\TALBOidub(C)$
we choose an appropriate branch~$\branch{B}$
and construct a model $\J$ which 
differs from $\I$
only in the interpretations of individuals occurring in~$\branch{B}$
but not occurring in $C$,
and so that
$\J\models\indiv:D$ for every concept $\indiv:D$ from $\branch{B}$.

More precisely, 
let $\Delta^\J\define\Delta^\I$ and
$\cname^\J\define \cname^\I$ and $\rname^\J\define \rname^\I$
for every concept symbol $\cname$ and role symbol~$\rname$.
Additionally $\indiv^\J\define\indiv^\I$ for every individual $\indiv$
which occurs in $C$.
In $\TALBOidub(C)$,
we find a branch $\branch{B}$ 
as a sequence of nodes $\seg{S}_0,\ldots,\seg{S}_n,\ldots$
and, for each~$n$,  we
extend interpretation $\J$ to individuals in $\seg{S}_n$
such that $\J\models\indiv:D$ for every concept $\indiv:D$ from $\seg{S}_n$.
It will follow immediately by the induction argument that
$\J\models\indiv:D$ for every concept $\indiv:D$ from $\branch{B}$.
\begin{itemize}
 \item For the base case of the induction argument
        we choose the root node of $\TALBOidub(C)$ as the initial node $\seg{S}_0$ of the branch $\branch{B}$.
	That is, $\seg{S}_0 = \{\indiv_0:C\}$ where $\indiv_0$ is a new individual.
        Since $C$ is satisfiable in $\I$, there is $w$ in $C^\J$.
        We let $\indiv_0^\J\define w$.
 \item The induction step involves considering cases for each of the rules in the $\TALBOidub$ calculus.
       We consider only two rules, namely \eqref{rule: exists} and~\eqref{rule: unrestricted blocking} rules.
       \begin{description}
        \item[\eqref{rule: exists}] Suppose $\indiv:\exists R.D$ is in the node $\seg{S}_n$
              and $\indiv:\exists R.\{\indiv'\}$ and $\indiv':D$ are the conclusions of the rule. 
              We set $\seg{S}_{n+1}\define\seg{S}_n\cup\{\indiv:\exists R.\{\indiv'\}, \indiv':D\}$.
              By the induction hypothesis we have that $\J\models\indiv:\exists R.D$.
              That is, there is an element $w$ in $\Delta^\J$ such that 
              $(\indiv^\J,w)\in R^\J$ and $w\in D^\J$.
              We set ${\indiv'}^\J\define w$.
        \item[\eqref{rule: unrestricted blocking}]
              Suppose $\indiv:\{\indiv\}$ and $\indiv':\{\indiv'\}$ are in the node $\seg{S}_n$.
              By the induction hypothesis $\indiv^{\J}$ and ${\indiv'}^{\J}$ are both defined.
              We have two cases: either $\indiv^{\J}={\indiv'}^{\J}$ or $\indiv^\J\neq {\indiv'}^{\J}$.
              We let $\seg{S}_{n+1}\define\seg{S}_n\cup\{\indiv:\{\indiv'\}\}$ in the first case and
                $\seg{S}_{n+1}\define\seg{S}_n\cup\{\indiv:\Not\{\indiv'\}\}$ in the second case.
       \end{description}
\end{itemize}
Thus, we constructed an appropriate branch $\branch{B}$.
In order to complete the construction of~$\J$ we set 
$\indiv^\J\define\indiv^\I$ for every individual $\indiv$
which does not occur in the constructed branch~$\branch{B}$.
Clearly, the model $\J$ satisfies the required properties.
$\branch{B}$ is open since  the set of all concepts from $\branch{B}$ is satisfiable in the model $\J$.
Now let $f$ be a function mapping $\Delta^{\J}$ ($=\Delta^\I$) to $\Delta^\IB$
and be defined by:
\[
f(w)\define%
 \begin{cases}
  \ecl{\indiv}, &\text{if $w=\indiv^{\J}$ for some individual $\indiv$ in $\branch{B}$},\\
  \text{arbitrary in $\Delta^\IB$}, & \text{otherwise}.
 \end{cases}
\]
By construction of $\J$
if $w=\indiv^{\J}={\indiv'}^{\J}$ for two different individuals $\indiv$ and $\indiv'$ 
from the branch $\branch{B}$
then $\indiv:\{\indiv'\}\in\branch{B}$ and, hence,
$\ecl{\indiv}=\ecl{\indiv'}$.
The function $f$ is thus defined correctly.
Furthermore, %
$f$ is onto~$\Delta^\IB$ because 
$\Delta^\IB=\{\ecl{\indiv}\mid \indiv:\{\indiv'\}\in\branch{B}\}$.
\end{proof}

Recall the model bounding function $\mu(n)$ obtained 
for $\ALBOid$ in Theorem~\ref{theorem: EFMP}: $\mu(n)=3(n\lfloor\log(n+1)\rfloor)\cdot 2^n$.
Combining Lemma~\ref{lemma: some magic} with Theorem~\ref{theorem: EFMP} %
we obtain the following theorem.
\begin{theorem}\label{theorem: finite branch existence}
Let $C$ be a satisfiable concept and $n$ be the length of $C$. %
Then
in every $\TALBOidub$ tableau for $C$,
there exists an open branch $\branch{B}$
such that $\Card(\Delta^\IB)\leq\mu(n)$.
\end{theorem}

\begin{theorem}[Termination]
 \TALBOidub is a terminating tableau calculus for satisfiability in \ALBOid.
\end{theorem}
\begin{proof}
In any $\TALBOidub$-tableau, every closed branch is
finite.
Suppose there is an open branch in the tableau.
By the Completeness Theorem~\ref{theorem: completeness} this means that the input concept is satisfiable.
By Theorem~\ref{theorem: finite branch existence} and Corollary~\ref{corollary: branch finiteness criterion}
there is a finite open branch in the tableau.
\end{proof}

Notice that
Condition~(c\ref{req: ub before exists}) is essential for ensuring termination
of a \TALBOidub derivation.
Figure~\ref{figure: req: ub before exists}
shows that 
without~(c\ref{req: ub before exists}) the 
\TALBOidub tableau for the concept 
\[\Not\exists(Q'\Or\Not Q').\Not\exists \rname.\cname\]
would not terminate because
new individuals are generated more often than the equality conjectures
made via the~\eqref{rule: unrestricted blocking}
rule.

\begin{figure}[tbu]
 \begin{center}
  \begin{minipage}{.75\textwidth}
  \begin{enumerate}[1.]
   \item\label{ex: c3: 0}
        \titem{\indiv_0:\Not\exists(Q'\Or\Not Q').\Not\exists \rname.\cname}{given}
   \item\label{ex: c3: 0a}
        \titem{\indiv_0:\{\indiv_0\}}{\eqref{rule: id},\ref{ex: c3: 0}}
   \item\label{ex: c3: 1}
        \titem{\indiv_0:\Not\exists Q'.\Not\exists \rname.\cname}{\eqref{rule: not exists or},\ref{ex: c3: 0}}
   \item\label{ex: c3: 2}
        \titem{\indiv_0:\Not\exists\Not Q'.\Not\exists \rname.\cname}{\eqref{rule: not exists or},\ref{ex: c3: 0}}
   \item\label{ex: c3: 3}
        \titem{\tbranch\indiv_0:\exists Q'.\{\indiv_0\}}{\eqref{rule: not exists not},\ref{ex: c3: 2},\ref{ex: c3: 0a}}
   \item\label{ex: c3: 3.1}
        \titem{\tskip\indiv_0:\Not\Not\exists \rname.\cname}{\eqref{rule: not exists},\ref{ex: c3: 1},\ref{ex: c3: 3}}
   \item\label{ex: c3: 3.2}
        \titem{\tskip\indiv_0:\exists \rname.\cname}{\eqref{rule: not not},\ref{ex: c3: 3.1}}
   \item\label{ex: c3: 3.3}
        \titem{\tskip\indiv_0:\exists \rname.\{\indiv_1\}}{\eqref{rule: exists},\ref{ex: c3: 3.2}}
   \item\label{ex: c3: 3.4}
        \titem{\tskip\indiv_1:\cname}{\eqref{rule: exists},\ref{ex: c3: 3.2}}
   \item\label{ex: c3: 3.5}
        \titem{\tskip\indiv_1:\{\indiv_1\}}{\eqref{rule: id},\ref{ex: c3: 3.4}}
   \item\label{ex: c3: 3.6}
        \titem{\tskip\tbranch\indiv_0:\exists Q'.\{\indiv_1\}}{\eqref{rule: not exists not},\ref{ex: c3: 2},\ref{ex: c3: 3.5}}
   \item\label{ex: c3: 3.6.1}
        \titem{\tskip[2]\indiv_1:\Not\Not\exists \rname.\cname}{\eqref{rule: not exists},\ref{ex: c3: 1},\ref{ex: c3: 3.6}}
   \item\label{ex: c3: 3.6.2}
        \titem{\tskip[2]\indiv_1:\exists \rname.\cname}{\eqref{rule: not not},\ref{ex: c3: 3.6.1}}
   \item\label{ex: c3: 3.6.3}
        \titem{\tskip[2]\indiv_1:\exists \rname.\{\indiv_2\}}{\eqref{rule: exists},\ref{ex: c3: 3.6.2}}
   \item\label{ex: c3: 3.6.4}
        \titem{\tskip[2]\indiv_2:\cname}{\eqref{rule: exists},\ref{ex: c3: 3.6.2}}
   \item\label{ex: c3: 3.6.5}
        \titem{\tskip[2]\indiv_2:\{\indiv_2\}}{\eqref{rule: id},\ref{ex: c3: 3.4}}
   \item\label{ex: c3: 3.6.6}   
        \titem{\tskip[2]\tbranch\indiv_0:\{\indiv_1\}}{\eqref{rule: unrestricted blocking},\ref{ex: c3: 0a},\ref{ex: c3: 3.5}}
   \item[\ldots]\titem{\tskip[3]\textbf{Non-terminating}}{Similarly to~\ref{ex: c3: 3.6}--\ref{ex: c3: 3.6.6}} 
   \item[]\titem{\tskip[3]{}}{with $\indiv_2$ in place of $\indiv_1$} 
   \item\label{ex: c3: 3.6.7}   
        \titem{\tskip[2]\tbranch\indiv_0:\Not\{\indiv_1\}}{\eqref{rule: unrestricted blocking},\ref{ex: c3: 0a},\ref{ex: c3: 3.5}}
   \item[\ldots]\titem{\tskip[3]{}}{} 
   \item\label{ex: c3: 3.7}
        \titem{\tskip\tbranch\indiv_1:\Not\Not\exists \rname.\cname}{\eqref{rule: not exists not},\ref{ex: c3: 2},\ref{ex: c3: 3.5}} 
   \item[\ldots]\titem{\tskip[2]{}}{} 
   \item\label{ex: c3: 4}
        \titem{\tbranch\indiv_0:\Not\Not\exists \rname.\cname}{\eqref{rule: not exists not},\ref{ex: c3: 2},\ref{ex: c3: 0a}} 
   \item[\ldots]\titem{\tskip{}}{}
  \end{enumerate}
  \end{minipage}
 \end{center}
\caption{Effect of omitting Condition~(c\ref{req: ub before exists})}\label{figure: req: ub before exists}
\end{figure}
In addition, omitting Condition~(c\ref{req: blocking}) %
causes any tableau for the concept $\indiv:\exists \rname.\cname$ not to terminate.
Thus, the conditions~(c\ref{req: blocking}) and~(c\ref{req: ub before exists}) are both necessary
conditions for termination of the calculus \TALBOidub. 

\section{A Complexity-Optimal Tableau Decision Procedure}
\label{section: Complexity}

Our estimations, in this section, for the complexity of tableau procedures based on the $\TALBOidub$ calculus
rely on the following two observations.
First, from Theorem~\ref{theorem: finite branch existence} 
it follows that, in a decision procedure based on $\TALBOidub$,
we can ignore any branch~$\branch{B}$ where the cardinality of $\Delta^\IB$ is larger than $\mu(n)$.
Any such branch is either closed or there is an open branch
with a smaller number of classes of equal individuals.

Second, Lemmas~\ref{lemma: concept number in branch} and~\ref{lemma: finite number of nominals} provide
a clue for estimating the number of derivation steps (that is, the number of rule applications)
in a shortest %
open branch $\branch{B}$.
Since the maximal number of premises of a rule is two,
and none of the rules are applied to the same set of premises more than once,
if $N$ is the number of concepts in $\branch{B}$, then the number of derivation steps in $\branch{B}$ is less than or equal to $N^2$.
Thus, the following lemma holds.
\begin{lemma}\label{lemma: number of derivation steps}
Let $\branch{B}$ be a finite open branch
in $\TALBOidub(C)$.
Then the number of derivation steps in $\branch{B}$ does not exceed
\begin{align}
\tag{$*$}
\left(n\cdot\left(k+M(\branch{B})\cdot\mu(n)\right)+2n\cdot\left(k+M(\branch{B})\cdot\mu(n)\right)^2\right)^2.
\end{align}
\end{lemma}

In order to define a complexity-optimal tableau decision procedure, let
us define the following derivation strategy.
\begin{quote}
\emph{Avoid huge branch strategy}:
Limit the number of derivation steps in every branch 
of a tableau derivation %
to the bound~($*$) in Lemma~\ref{lemma: number of derivation steps}
and do not expand branches that become longer.
\end{quote}

Whereas $n$ and $k$ depend only on the given concept $C$, the value
$M(\branch{B})$ %
depends
also on the particular strategy of applying the~\eqref{rule:
unrestricted blocking} rule.
In fact, the following lemma holds.

\begin{lemma}
In a branch $\branch{B}$
let $M'$ be the number of individuals occurring
strictly before the node from which point onwards, before any application of
the \eqref{rule: exists}~rule,
all applications of the \eqref{rule: unrestricted blocking}~rule
have been performed in accordance with 
Condition~(c\ref{req: ub before exists}).
Let $n'$ be the number of concepts of the form $\exists R.D$ in $\sub(C)$.
Then $M(\branch{B})\leq M'+n'$.
\end{lemma}

\begin{proof}[(sketch)]
It can be seen from the proof of Lemma~\ref{lemma: napplE(ell)}
that, after the node where exhaustive
applications of the \eqref{rule: unrestricted blocking}~rule started,
the number of applications of the \eqref{rule: exists}~rule
to concepts labelled by individuals from the same equivalence class
in the branch $\branch{B}$ 
is restricted by $n'$.
The number of applications of the \eqref{rule: exists}~rule to concepts
before this node is restricted by the number of individuals which
are present before that node, that is, by~$M'$.
\end{proof}
We note that $n'\leq n$, but, depending on the strategy used, $M'$ can have any value. 
However, if a tableau strategy ensures that $M'$ is at most exponential
in the length of the given concept $C$
then the number of derivation steps in the branch $\branch{B}$
is also at most exponential in the length of $C$.
This is in particular true
if Condition~(c\ref{req: ub before exists})
is satisfied for the \emph{first} node of the tableau.
In this case, $M'=0$. 
This proves the following theorem.

\begin{theorem}
The calculus \TALBOidub provides a 
non-deterministic tableau procedure running in \NExpTime,
when it uses 
the `avoid huge branch strategy' together with
the tableau expansion strategy ensuring that before
any application of the \eqref{rule: exists}~rule
all possible applications of the \eqref{rule: unrestricted blocking}~rule
have been performed.
\end{theorem}

This gives a complexity optimal tableau decision procedure for
\ALBOid.

\section{Deterministic Decision Procedures}
\label{section_implementation}

We have presented a sound, complete and terminating tableau calculus for \ALBOid.
The calculus defines the rules, we have defined how a tableau
derivation may be constructed and how rules may be applied, but it has
not been specified in concrete terms how the tableau derivation must
be constructed and how the the rules must be applied.
The presented tableau calculi define only non-deterministic decision procedures.

When implementing a calculus as a deterministic
decision procedure an important issue therefore is to decide how to perform the
search without loosing the crucial properties of soundness, completeness
and termination. As all the rules in the calculus \TALBOidub are sound, preserving
soundness is not problematic.
Care is needed for preserving completeness and termination.
In particular, it is crucial that the search is performed in a fair way.
A procedure is \emph{fair} if, when an inference is possible forever,
then it is performed eventually.
This means a deterministic tableau procedure
based on \TALBOidub
may not defer the use of an applicable rule indefinitely.
In our context fairness must be understood in a \emph{global} sense.
That is, a tableau procedure has to be fair not only to expressions in
a particular branch but to expressions in \emph{all} branches of a tableau.
In other words, a procedure is fair if it makes both the branch
selection, and the selection of expressions to which to apply a rule,
in a fair way.
Then soundness, completeness and termination carry over from the
calculus and we have:

\begin{theorem}
Any fair tableau procedure based on \TALBOidub is a decision procedure for
\ALBOid and all its sublogics.
\end{theorem}

\begin{figure}
\begin{center}
\begin{minipage}{.8\textwidth}
 \begin{enumerate}[1.]%
  \item\label{ex: unsat: 0}
    \titem{\indiv_0:\Not\left(\exists (Q'\Or\Not Q').\Not\exists \rname.\cname\Or\Not\exists Q''.\Not\exists \rname.\cname\right)}{given}
  \item\label{ex: unsat: 0a}
    \titem{\indiv_0:\{\indiv_0\}}{\eqref{rule: id},\ref{ex: unsat: 0}}
  \item\label{ex: unsat: 1}\titem{\indiv_0:\Not\exists (Q'\Or\Not Q').\Not\exists \rname.\cname}{\eqref{rule: not or},\ref{ex: unsat: 0}}
  \item\label{ex: unsat: 2a}
    \titem{\indiv_0:\Not\Not\exists Q''.\Not\exists \rname.\cname}{\eqref{rule: not or},\ref{ex: unsat: 0}}
  \item\label{ex: unsat: 2}
    \titem{\indiv_0:\exists Q'' .\Not\exists \rname.\cname}{\eqref{rule: not not},\ref{ex: unsat: 2a}}
  \item\label{ex: unsat: 3a}
    \titem{\indiv_0:\Not\exists Q'.\Not\exists \rname.\cname}{\eqref{rule: not exists or},\ref{ex: unsat: 1}}
  \item\label{ex: unsat: 3b}
    \titem{\indiv_0:\Not\exists\Not Q'.\Not\exists \rname.\cname}{\eqref{rule: not exists or},\ref{ex: unsat: 1}}
  \item\label{ex: unsat: 3b.1}
    \titem{\tbranch\indiv_0:\exists Q'.\{\indiv_0\}}{\eqref{rule: not exists not},\ref{ex: unsat: 3b},\ref{ex: unsat: 0a}}
  \item\label{ex: unsat: 3.2}
    \titem{\tskip\indiv_0:\Not\Not\exists \rname.\cname}{\eqref{rule: not exists},\ref{ex: unsat: 3b.1},\ref{ex: unsat: 3a}}
  \item\label{ex: unsat: 3.3}
    \titem{\tskip\indiv_0:\exists \rname.\cname}{\eqref{rule: not not},\ref{ex: unsat: 3.2}}
  \item\label{ex: unsat: 3.4}\titem{\tskip\indiv_1:\cname}{\eqref{rule: exists},\ref{ex: unsat: 3.3}}
  \item\label{ex: unsat: 3.5}\titem{\tskip\indiv_0:\exists \rname.\{\indiv_1\}}{\eqref{rule: exists},\ref{ex: unsat: 3.3}}
  \item\label{ex: unsat: 3.4a}
    \titem{\tskip\indiv_1:\{\indiv_1\}}{\eqref{rule: id},\ref{ex: unsat: 3.4}}

  \item\label{ex: unsat: 3.6}
    \titem{\tskip\tbranch\indiv_0:\exists Q'.\{\indiv_1\}}{\eqref{rule: not exists not},\ref{ex: unsat: 3b},\ref{ex: unsat: 3.4a}}
  \item\label{ex: unsat: 3.6.1}
    \titem{\tskip[2]\indiv_1:\Not\Not\exists \rname.\cname}{\eqref{rule: not exists},\ref{ex: unsat: 3.6},\ref{ex: unsat: 3a}}
  \item\label{ex: unsat: 3.6.2}
    \titem{\tskip[2]\indiv_1:\exists \rname.\cname}{\eqref{rule: not not},\ref{ex: unsat: 3.6.1}}

  \item\label{ex: unsat: 3.6.3}\titem{\tskip[2]\tbranch \indiv_0:\{\indiv_1\}}{\eqref{rule: unrestricted blocking},\ref{ex: unsat: 0a},\ref{ex: unsat: 3.4a}}

  \item\label{ex: unsat: 3.6.3.1}
    \titem{\tskip[3]\indiv_2:\Not\exists \rname.\cname}{\eqref{rule: exists},\ref{ex: unsat: 2}}
  \item\label{ex: unsat: 3.6.3.2}
    \titem{\tskip[3]\indiv_0:\exists Q''.\{\indiv_2\}}{\eqref{rule: exists},\ref{ex: unsat: 2}}
  \item\label{ex: unsat: 3.6.3.1a}
    \titem{\tskip[3]\indiv_2:\{\indiv_2\}}{\eqref{rule: id},\ref{ex: unsat: 3.6.3.1}}
  
  \item\label{ex: unsat: 3.6.3.3}
    \titem{\tskip[3]\tbranch\indiv_0:\exists Q'.\{\indiv_2\}}%
    {\eqref{rule: not exists not},\ref{ex: unsat: 3b},\ref{ex: unsat: 3.6.3.1a}}

  \item\label{ex: unsat: 3.6.3.3.1}
    \titem{\tskip[4]\indiv_2:\Not\Not\exists \rname.\cname}%
    {\eqref{rule: not exists},\ref{ex: unsat: 3a},\ref{ex: unsat: 3.6.3.3}}
  \item\label{ex: unsat: 3.6.3.3.2}
    \titem{\tskip[4]\unsat}{\eqref{rule: clash},\ref{ex: unsat: 3.6.3.1},\ref{ex: unsat: 3.6.3.3.1}} 

  \item\label{ex: unsat: 3.6.3.4.1}
    \titem{\tskip[3]\tbranch\indiv_2:\Not\Not\exists \rname.\cname}%
    {\eqref{rule: not exists not},\ref{ex: unsat: 3b},\ref{ex: unsat: 3.6.3.1a}}

  \item\label{ex: unsat: 3.6.3.4.2}
    \titem{\tskip[4]\unsat}{\eqref{rule: clash},\ref{ex: unsat: 3.6.3.1},\ref{ex: unsat: 3.6.3.4.1}} 

  \item\label{ex: unsat: 3.6.4}\titem{\tskip[2]\tbranch \indiv_0:\Not\{\indiv_1\}}{\eqref{rule: unrestricted blocking},\ref{ex: unsat: 0a},\ref{ex: unsat: 3.4a}}
  \item\label{ex: unsat: 3.6.4.1}\titem{\tskip[3] \indiv_2:\cname}%
    {\eqref{rule: exists},\ref{ex: unsat: 3.6.2}}  
  \item\label{ex: unsat: 3.6.4.2}\titem{\tskip[3] \indiv_1:\exists \rname.\{\indiv_2\}}%
    {\eqref{rule: exists},\ref{ex: unsat: 3.6.2}}  
  \item\label{ex: unsat: 3.6.4.1a}
    \titem{\tskip[3]\indiv_2:\{\indiv_2\}}{\eqref{rule: id},\ref{ex: unsat: 3.6.4.1}}
  \item\label{ex: unsat: 3.6.4.3}
    \titem{\tskip[3]\tbranch\indiv_0:\exists Q'.\{\indiv_2\}}%
    {\eqref{rule: not exists not},\ref{ex: unsat: 3b},\ref{ex: unsat: 3.6.4.1a}}  
  \item\label{ex: unsat: 3.6.4.3.1}
    \titem{\tskip[4]\indiv_2:\Not\Not\exists \rname.\cname}%
    {\eqref{rule: not exists},\ref{ex: unsat: 3.6.4.3},\ref{ex: unsat: 3a}}  
  \item\label{ex: unsat: 3.6.4.3.2}
    \titem{\tskip[4]\indiv_2:\exists \rname.\cname}%
    {\eqref{rule: not not},\ref{ex: unsat: 3.6.4.3.1}}
  \item[\ldots]\label{ex: unsat: 3.6.4.3.3}
    \titem{\tskip[4]\textbf{Non-terminating}}%
    {Repetition of~\ref{ex: unsat: 3.6.2}--\ref{ex: unsat: 3.6.4.3.2}}
  
  \item\label{ex: unsat: 3.6.4.4}
    \titem{\tskip[3]\tbranch\indiv_2:\Not\Not\exists \rname.\cname}%
    {\eqref{rule: not exists not},\ref{ex: unsat: 3b},\ref{ex: unsat: 3.6.4.1a}}  
  \item\label{ex: unsat: 3.6.4.4.1}
    \titem{\tskip[4]}
    {Similarly to~\ref{ex: unsat: 3.6.4.3}--\ref{ex: unsat: 3.6.4.3.3}}
  
  \item\label{ex: unsat: 3.7}
    \titem{\tskip\tbranch\indiv_1:\Not\Not\exists \rname.\cname}%
        {\eqref{rule: not exists not},\ref{ex: unsat: 3b},\ref{ex: unsat: 3.4a}}
  \item\label{ex: unsat: 3.7.1}
    \titem{\tskip[2]}
    {Similarly to~\ref{ex: unsat: 3.6}--\ref{ex: unsat: 3.6.4.4.1}}
  
  \item\label{ex: unsat: 3b.2}
    \titem{\tbranch\indiv_0:\Not\Not\exists \rname.\cname}%
    {\eqref{rule: not exists not},\ref{ex: unsat: 3b},\ref{ex: unsat: 0a}}
  \item\label{ex: unsat: 3b.2.1}
    \titem{\tskip}
    {Similarly to~\ref{ex: unsat: 3b.1}--\ref{ex: unsat: 3.7.1}}
 \end{enumerate}
\end{minipage}
\end{center}
\caption{An infinite derivation, due to unfair selection of concepts}
\label{fig_example_fairness}
\end{figure}%

To illustrate the importance of fairness we give two examples.
The concept
\[
\Not\Bigl(\exists (Q'\Or\Not Q').\Not\exists \rname.\cname\Or\Not\exists Q''.\Not\exists \rname.\cname\Bigr),
\]
or equivalently
$(\Box\exists \rname.\cname)\And(\exists Q''.\forall \rname.\Not \cname)$,
is not satisfiable.
Figure~\ref{fig_example_fairness}
gives a depth-first left-to-right
derivation that is unfair and does not terminate.
It can be seen that the derivation is infinite because the application
of the \eqref{rule: exists}~rule to $\indiv_0:\exists Q''.\Not\exists \rname.\cname$
is deferred forever and, consequently, a contradiction is not found.
This illustrates the importance of fairness for 
completeness.

\begin{figure}[tbu]
\begin{center}
\begin{minipage}{.725\textwidth}
 \begin{enumerate}[1.]
  \item\label{ex: sat: 0}\titem{\indiv_0:\Not\exists(Q'\Or\Not Q').\Not\exists \rname.\cname}{given}
  \item\label{ex: sat: 0a}\titem{\indiv_0:\{\indiv_0\}}{\eqref{rule: id},\ref{ex: sat: 0}}
  \item\label{ex: sat: 1}\titem{\indiv_0:\Not\exists Q'.\Not\exists \rname.\cname}{\eqref{rule: not exists or},\ref{ex: sat: 0}}
  \item\label{ex: sat: 2}\titem{\indiv_0:\Not\exists\Not Q'.\Not\exists \rname.\cname}{\eqref{rule: not exists or},\ref{ex: sat: 0}}
  \item\label{ex: sat: 3.0}\titem{\tbranch\indiv_0:\Not\Not\exists \rname.\cname}%
    {\eqref{rule: not exists not},\ref{ex: sat: 2},\ref{ex: sat: 0a}}
  \item\label{ex: sat: 3.1}\titem{\tskip\indiv_0:\exists \rname.\cname}{\eqref{rule: not not},\ref{ex: sat: 3.0}}
  \item\label{ex: sat: 3.2}\titem{\tskip\indiv_1:\cname}{\eqref{rule: exists},\ref{ex: sat: 3.1}}
  \item\label{ex: sat: 3.3}\titem{\tskip\indiv_0:\exists \rname.\{\indiv_1\}}{\eqref{rule: exists},\ref{ex: sat: 3.1}}
  \item\label{ex: sat: 3.2a}\titem{\tskip\indiv_1:\{\indiv_1\}}{\eqref{rule: id},\ref{ex: sat: 3.2}}
  \item\label{ex: sat: 3.4.0}\titem{\tskip\tbranch\indiv_1:\Not\Not\exists \rname.\cname}%
    {\eqref{rule: not exists not},\ref{ex: sat: 2},\ref{ex: sat: 3.2a}}
  \item\label{ex: sat: 3.4.1}\titem{\tskip[2]\indiv_1:\exists \rname.\cname}{\eqref{rule: not not},\ref{ex: sat: 3.0}}
  \item\label{ex: sat: 3.4.2}\titem{\tskip[2]\tbranch\indiv_0:\Not\{\indiv_1\}}%
    {\eqref{rule: unrestricted blocking},\ref{ex: sat: 0a},\ref{ex: sat: 3.2a}}
  \item\label{ex: sat: 3.4.2.0}\titem{\tskip[3]\indiv_2:\cname}{\eqref{rule: exists},\ref{ex: sat: 3.4.1}}
  \item\label{ex: sat: 3.4.2.1}\titem{\tskip[3]\indiv_1:\exists \rname.\{\indiv_2\}}%
    {\eqref{rule: exists},\ref{ex: sat: 3.4.1}}
  \item[\ldots]\label{ex: sat: 3.4.2.2}
    \titem{\tskip[3]\textbf{Non-terminating}}%
    {Repetition of~\ref{ex: sat: 3.2}--\ref{ex: sat: 3.4.2.1}}
  \item\label{ex: sat: 3.4.3}\titem{\tskip[2]\tbranch\indiv_0:\{\indiv_1\}}%
    {\eqref{rule: unrestricted blocking},\ref{ex: sat: 0a},\ref{ex: sat: 3.2a}}
  \item[\ldots]\label{ex: sat: 3.4.3.0}
    \titem{\tskip[3]}
    {Never expanded}
  \item\label{ex: sat: 3.5}\titem{\tskip\tbranch\indiv_0:\exists Q'.\{\indiv_1\}}%
    {\eqref{rule: not exists not},\ref{ex: sat: 2},\ref{ex: sat: 3.2a}}
  \item[\ldots]\label{ex: sat: 3.5.0}
    \titem{\tskip[2]}
    {Never expanded}
  \item\label{ex: sat: 4}\titem{\tbranch\indiv_0:\exists Q'.\{\indiv_0\}}%
    {\eqref{rule: not exists not},\ref{ex: sat: 2},\ref{ex: sat: 0a}}
  \item[\ldots]\label{ex: sat: 4.0}
    \titem{\tskip}
    {Never expanded}
 \end{enumerate}
\end{minipage}
\end{center}
\caption{An infinite derivation, due to unfair selection of branches}
\label{figure: unfair branch choice}
\end{figure}

The next example illustrates the importance of fairness for termination.
The concept
\[
\Not\exists(Q'\Or\Not Q').\Not\exists \rname.\cname,
\]
or equivalently 
$\Box\exists \rname.\cname$,
is satisfiable.
The derivation in Figure~\ref{figure: unfair branch choice} is
obtained with a depth-first \emph{right-to-left} strategy.
However, the repeated selection of the right branch at 
\eqref{rule: unrestricted blocking}~choice points, in particular,
causes all the individuals in the branch to be pair-wise non-equal.
The concept $\indiv:\exists \rname.\cname$ re-appears repeatedly
for every individual $\indiv$ in the branch.
This triggers the repeated generation of a new individual by the 
\eqref{rule: exists}~rule, resulting in an infinite
derivation.
This strategy is unfair because all branches except for the right-most
branch get ignored.
In the branches right-most with respect to the \eqref{rule: unrestricted blocking}~rule, it is as if blocking with the \eqref{rule: unrestricted blocking}~rule
has never been applied.

Recall from Section~\ref{section_tableau}, a tableau calculus is terminating iff the tableau
constructed for any finite concept set $\mathcal{S}$ is finite whenever $T(\mathcal{S})$ is closed or
contains a finite open
branch if the tableau is open.
For a satisfiable concept set this means it is possible to construct
\emph{a} finite, fully expanded, open branch.
Not all branches in an open terminating tableau need to be finite though.
\citeN{Reker-PhD-2011} gives an example transferable to \ALBOid
for which the left-most open branch in a tableau derivation is infinite.

An immediate way of obtaining a fair deterministic decision
procedure for a terminating tableau is to use breadth-first search.
Another possibility is to use depth-first search with iterative
deepening.
The idea of iterative deepening search is that the tableau is
expanded up to a fixed depth. 
If a closed tableau is found or if an open fully expanded branch is
found then the derivation process stops.
If not, then the depth is increased and the tableau is constructed up to
this depth, and so on.
A benefit of breadth-first and depth-first iterative deepening
strategies is the generation of small models.

Additionally, depth-first search combined with the `avoid huge branch
strategy' provides a fair, deterministic decision procedure.

\begin{theorem}
The following are deterministic decision procedures for
\ALBOid and all its sublogics:
\begin{enumerate}
 \item 
Any fair tableau procedure based on \TALBOidub using a breadth-first 
search strategy.
\item 
Any fair tableau procedure based on \TALBOidub using a depth-first search strategy
with iterative deepening.
 \item 
Any fair tableau procedure based on \TALBOidub using a depth-first
left-to-right strategy and the `avoid huge branch strategy'.%
\end{enumerate}
\end{theorem}

Whether a depth-first left-to-right strategy without exploiting the model bound
gives a decision procedure for \ALBOid  is open.
It is in particular open whether there are strategies to force the left-most
open branch to be finitely bounded.

\section{Discussion}
\label{section_discussion}

An important contribution of this paper is the unrestricted blocking mechanism
based on the unrestricted blocking rule.
It is conceptually similar but more generic than the approach of
reusing domain terms in, for example, \cite{BryManthey87,BryTorge98}, where it is
used to compute domain minimal models.
The following rule
\[
 \tableaurule{\indiv:\exists R.C}{\indiv:\exists R.\{\indiv_0\}\tand\indiv_0:C\tor\cdots\tor\indiv:\exists R.\{\indiv_n\}\tand\indiv_n:C%
  \tor\indiv:\exists R.\{\indiv'\}\tand\indiv':C},
\]
where $\indiv_0,\ldots,\indiv_n$ are all the individuals occurring in the current branch 
and $\indiv'$ is a fresh individual,
is an adaptation of the $\delta^\ast$-rule
from Bry et al.\ \citeyear{BryManthey87,BryTorge98} 
to the language of description logics.
Every tableau derivation
with this rule can be rewritten to an equivalent tableau derivation
using the unrestricted blocking rule
but not the other way around.

An advantage of the unrestricted blocking mechanism is that
it allows to separate proofs of termination of tableau calculi from
proofs of their soundness and completeness.
In the description, modal and hybrid logic literature blocking mechanisms are
usually tightly integrated with the tableau rules and the tableau procedure.
Soundness and completeness of the procedure is proved taking into
account the given blocking mechanism.
This creates intricacies and conceptual dependencies
in the proofs.
In contrast, blocking based on the \eqref{rule: unrestricted blocking}~rule is not tied to a specific logic.
It can be added to any sound and complete
tableau calculus without loosing soundness and completeness.

Exploiting this we have shown that
\emph{any} tableau procedure based on a tableau calculus extended
with this blocking mechanism is guaranteed to terminate if some
additional general conditions are satisfied~\cite{SchmidtTishkovsky-GTM+-2008}.
These conditions are the blocking conditions~(c\ref{req: blocking})
and~(c\ref{req: ub before exists}) and fairness of a particular tableau
strategy used in the tableau procedure, which is essential for
the method to work.
In addition, the effective finite model property must hold and it must
be proved by a filtration argument correlating in an appropriate
way with the original tableau calculus.
The present paper provides a different kind of proof for termination
to that in~\cite{SchmidtTishkovsky-GTM+-2008}.
The proof here can be extended to show
termination of tableau calculi equipped with the unrestricted blocking mechanism
for logics for which the effective model property may not hold but
the \emph{finite model property} still holds.
Proving the finite model property can be involved but 
proofs of it 
may already be known.
This means %
that known finite model
property results can be readily exploited to define terminating tableau
calculi.
We have introduced methods for systematically developing semantic tableau calculi for
modal and related logics \cite{Schmidt09d,SchmidtTishkovsky-ASTC-lmcs-2011}.
Adding unrestricted blocking turns many of the generated tableau
calculi into terminating
ones.

The way the blocking mechanism is defined in our calculus means
that from some point onwards it needs to be used eagerly, as is required
by Condition~(c\ref{req: ub before exists}).
This restriction still leaves a lot of room for different
tableau expansion strategies.
For example, a tableau procedure can postpone applications of the
\eqref{rule: unrestricted blocking}~rule or apply the rule in a
non-eager way while the provisional model constructed in a branch is not
too large.
Subtle measures and strategies can be devised (for instance, based
on run time or available memory) for determining when to start
applying the blocking rule eagerly.

Eager application of the blocking rule produces small models.
This is beneficial because small counter-examples can be produced
whereas models generated by standard blocking mechanisms are generally
larger.
With smaller models memory consumption
can be significantly smaller leading to quicker answers and improved
performance.

On the other hand the blocking rule involves branching and excessive
branching can degrade performance.
It can thus be advantageous to apply blocking only selectively.
Existing standard blocking mechanisms only identify individuals if
certain blocking conditions are true.
In particular, they are based on inclusion or equality tests on
sets of expressions that must normally be constructed in a particular
way. Individuals are identified only implicitly through the use
of status variables that identify individuals as blocked and blockable.
The expansion is usually assumed to be depth-first and the rules are
usually assumed to be applied in a certain order.
The unrestricted blocking mechanism is less restrictive and is monotone:
there is no blocking test, any individual is
blockable and blocked individuals remain blocked in a branch.
There are conditions for the application of the unrestricted
blocking rule but they are designed to be as weak and permissive
as possible.
As is shown in this paper none of the conditions can be omitted. 

Though standard blocking mechanisms are defined significantly
differently, they can be viewed as specialisations and adaptations of
unrestricted blocking.
Standard blocking amounts to the restricted application of blocking,
sometimes without change of the status of blocked individuals
(static blocking), sometimes with
change in the status of blocked individuals (dynamic
blocking).
Restricting applications of blocking reduces the number of branching
nodes and the overall search space can be significantly smaller.
A particular blocking test may need to be realised by an algorithm
complicated enough to create considerable overhead.
Ultimately, the trade-off between the resources required for blocking
tests and excessive branching created as a result of applications of
the blocking rule is inescapable and needs to be carefully
considered.

Different standard blocking mechanisms are needed for different logics.
While one blocking mechanism might work for a particular logic it can
be incomplete for a more expressive logic.
The advantage of unrestricted blocking is automatic soundness
and completeness for any logic.
Analysis of the termination proof in~\citeN{SchmidtTishkovsky-GTM+-2008} however suggests that,
in the case of a particular logic, an appropriate
blocking test can be extracted from the definition of a filtrated model
used in the proof of the effective finite model property for the logic.
This opens a promising direction for further research into
devising new blocking tests and blocking mechanisms
in a systematic way.

Compared to the tableau calculi of the conference paper~\cite{SchmidtTishkovsky-UTD+-2007}
there are a few differences in the calculus $T_\ALBOid$.
Besides the presence of rules for the identity role in $T_\ALBOid$,
the essential difference is that the calculi in the conference version 
include
this additional congruence rule:
\[
\tableaurule{\indiv':\{\indiv''\}\tand\indiv:\exists R.\{\indiv'\}
}{\indiv:\exists R.\{\indiv''\}}[bridge].\label{rule: bridge}
\]
The results of the present paper show that this rule is superfluous for
completeness of $T_\ALBOid$.
(We remark the \eqref{rule: bridge}~rule is sound but
not derivable in $\TALBOidub$.)

Our use of the equality constraints $\indiv:\{\indiv\}$ in the tableau rules
is non-standard and employs an idea
proposed and implemented in \mettel tableau
prover~\cite{mettel_page,HustadtTishkovskyWolterZakharyaschev-ARMT-2006,TishkovskySchmidtKhodadadi-MetTeL-2011}.
Normally these labelled concepts would be deleted by standard simplification rules
since $\indiv:\{\indiv\}$ is valid.
This would not be correct in our calculus because we use these equality
constraints as \emph{domain predicates} for keeping track of the
individuals \emph{essential} for instantiation by the $\gamma$~rules
in the calculus, and ultimately, for keeping track of the individuals
essential for the construction of the model in a tableau branch.
Here, by
$\gamma$~rules we mean the rules for universally quantified expressions
creating, by repeated application, all instantiations of individuals
occurring in conclusions but not in the premise (that is, the implicit
universally quantified first-order variables).
The $\gamma$~rules of the $\TALBOidub$ calculus are the
\eqref{rule: not exists not}~rule and the \eqref{rule: unrestricted blocking}~rule.
Naively these rules can be defined to use all individuals occurring in
the input set and all individuals introduced by the \eqref{rule: exists}~rule for
instantiation but this would create unnecessarily many instances and
would mean the search space becomes unnecessarily large.
The \eqref{rule: not exists}~rule does not require domain predication because
this is already achieved by the second premise of the rule.
Domain predication can be obtained in other ways by using an explicit
concept name representing the domain as is done
in~\citeN{HustadtSchmidt99b}, \citeN{Schmidt06Dagstuhl}, and \citeN{BaumgartnerSchmidt06}.
If preferred, domain predication can be omitted and captured by
appropriate side-conditions, but then the \eqref{rule: id}~rule and the
$\gamma$~rules need to be reformulated and the proofs adapted slightly.

Although the proofs in this paper are relatively concise,
they are not trivial.
They are based on the idea
that the %
finite model property for \ALBOid
and termination property of the proposed tableau calculus
are of the same nature, which is an important contribution.
In particular, it was proved
that a finite model (if it exists) %
can be generated by applying the unrestricted 
blocking rule eagerly within the branch.

While our tableau approach is in line with existing semantic tableau
approaches described in the description, modal and hybrid logic literature,
our presentation departs in significant ways from other presentations.
To make the approach and the results as general as possible, we have
presented it in terms of an abstract deduction calculus given by a set
of inference rules without making any assumptions about how branch
selection is performed, how tableau derivations are constructed,
in which order the rules are applied, which individuals are blockable and
the need for using status variables.
Another difference is that we view a tableau to be a (tree) derivation
rather than a model as is the tendency in the description logic
literature.
In our context models are in general not tree models or even tree-like,
because \ALBOid does not have the tree-model property. The
unrestricted blocking rule can identify any pair of individuals meaning
that any structural properties, like being a tree model, are easily lost.
This shift in perspective is useful for developing tableau approaches
for description logics without a form of tree model property.

Practical considerations such as branch selection order, rule
application order, search strategies and heuristics are 
important considerations when embarking on an implementation.
We touched on practical aspects and described general notions
of fairness providing minimal conditions to guarantee sound, complete
and terminating deterministic tableau procedures.
The standard optimisations such as simplification, backjumping,
dynamic backtracking, different heuristics for branch selection and
rule selection, lemma generation, et cetera, are all compatible with
the presented calculi and procedures.
An obvious simplification, for example, is the on-the-fly
removal of double negations from concepts, and especially from roles,
as this reduces a number of applications of the 
\eqref{rule: not exists not}~rule.

The unrestricted blocking mechanism
has been implemented in \mettel tableau
prover~\cite{mettel_page,TishkovskySchmidtKhodadadi-MetTeL-2011},
its successor, the prover generator \mettelII~\cite{TishkovskySchmidtKhodadadi-MetTeL2-2012a}, 
and the \mspass first-order resolution theorem prover.
Tests with various description logic problems and 
problems in the TPTP library show that the general termination
mechanism by means of the unrestricted blocking rule does work in
practice~\cite{BaumgartnerSchmidt08}.
A more thorough discussion of these implementations and other
practical aspects, let alone performance results, is beyond the scope
of the paper.

\section{Conclusion}

We have presented a new, general tableau approach for deciding
expressive description logics with complex role operators, including 
`non-safe' occurrences of role negation.
The tableau decision procedures found in the description logic
literature, and implemented in existing tableau-based description
logic systems, can handle a large class of description logics, but
cannot currently handle description logics with full role negation
such as \ALB, \ALBO or \ALBOid.
The present paper closes this gap.
The introduced blocking mechanism opens the way for improving the OWL 2 standard
to cover larger decidable fragments of first-order logic and larger fragments of natural language
than it currently does, with obvious benefits for applications of
the semantic web.

\begin{acks}
We thank Hilverd Reker for useful discussions.
The work was financially supported by research grants~EP/D056152/1, EP/F068530/1
and
EP/H043748/1 of the UK EPSRC.
\end{acks}

\end{document}